\documentclass[a4paper]{article}

\usepackage{fullpage,url,amsmath,amsfonts,amssymb, tedmath}
\newcommand{\tr}{\mbox{tr}}
\newcommand{\cc}{\mathbb{C}}

\newcommand{\cir}{\,\text{circ}}
\newcommand{\z}{\zeta}

\newcommand{\ov}{\overline}
\title{Paraunitary matrices}
\author{Barry Hurley\footnote{National University
 of Ireland Galway, email: Barryj\_2001@yahoo.co.uk} \, \& Ted
 Hurley\footnote{National Universiy of Ireland Galway, email:
 Ted.Hurley@Nuigalway.ie }}
\date{}

\begin{document}

\maketitle

\begin{centering} 


\end{centering}




\begin{abstract}

Design methods for paraunitary matrices from 
complete orthogonal sets of idempotents and related matrix structures
 are presented. 
These include techniques for designing  non-separable   
multidimensional paraunitary matrices. 
Properties of  the structures are obtained and proofs given. 
 Paraunitary matrices play a central role in signal processing, in
 particular in the areas of filterbanks and wavelets. 
\end{abstract}

\section{Introduction}

 A one-dimensional (1D) paraunitary matrix over $\cc$ is a square
matrix $U(z)$ satisfying 
$U(z){U}^*(z^{-1}) = 1$. Here $^*$ denotes complex conjugate transposed and $1$
denotes the identity matrix of the size of $U(z)$. 
In general a $k$-dimensional (kD) paraunitary
matrix over $\cc$ is a matrix $U(\bf{z})$ where ${\bf{z}}=(z_1,z_2,
\ldots, z_k)$ is 
a vector of (commuting) variables $\{z_1,z_2, \ldots, z_k\}$ such that
$U({\bf{z}}){U}^*({\bf{z}^{-1}}) = 1$ and  
 ${\bf{z}^{-1}} = (z_1^{-1},z_2^{-1}, \ldots, z_k^{-1})$.

Over fields other than $\cc$ a
 paraunitary matrix  is a matrix $U(\bf{z})$ satisfying 
 $U({\bf{z}}){U}({\bf{z}^{-1}})\T = 1$.   

Paraunitary matrices are important in signal
  processing and in particular the concept of a paraunitary matrix
  plays a fundamental role in the research area of multirate 
filterbanks and  wavelets.  
In the polyphase domain, the synthesis matrix of an orthogonal filter
bank is a paraunitary matrix;  see for example \cite{strang}. 

Orthogonal filter banks may also
be used to construct orthonormal wavelet bases \cite{daub1,daub2}; see
also references in \cite{zhou}. Paraunitary matrices over finite fields
have been studied for their own interest and for applications; see
for example \cite{finf}.

Here general methods for constructing and designing such
  matrices from complete orthogonal sets of idempotents together with 
related matrix schemes are presented. This includes methods for
  designing non-separable multidimensional paraunitary matrices. 
Construction methods for complete orthogonal sets of
idempotents are included. Group ring construction methods were the
original motivation and from these more general methods evolved.
A structure called
  the {\em tangle} of matrices is introduced; this may have independent
  interest.
  
 

In certain cases  specialising the
variables of the paraunitary matrices allows the construction of series of  
regular real or complex Hadamard matrices.
Walsh-Hadamard matrices, used extensively in the communications' areas, 
are examples of such regular Hadamard matrices. Complex Hadamard
matrices arise in the study of operator algebras and
 in the theory of quantum computation.

It is noted  
 that the  renowned building blocks for 1D
paraunitary matrices over $\cc$ due to
Belevitch and Vaidyanathan as described in
\cite{vaid} are constructed from $W=\{F_1,F_2\}$ where
$W$ is a complete orthogonal set of two idempotents in which $F_1$ has
  rank $1$ and $F_2$ has rank $(n-1)$   with  $n$ the size of the
matrices under consideration. See section \ref{tub} below for details on
 this.

Connections between group
rings, matrices and design of codes 
have been established in \cite{hur1}, \cite{hur2} and
\cite{hur3}; these are related but independent. 



Designing non-separable multidimensional paraunitary matrices 
is deemed difficult as 
there is no multidimensional factorisation theorem corresponding to the
1D factorisation theorem of Belevitch and Vaidyanathan (\cite{vaid}). 
For the finite impulse response (FIR) case there seems to be only a few
examples such as \cite{kova}. See also \cite{zhou} for
background and further discussion. 
In \cite{del} a factorization of a subclass of 2D paraunitary matrices
is obtained; these though  involve IIR (infinite impulse response) systems.

In Section \ref{grmat} results are obtained 
on the ranks of the idempotents  and on the determinants of the paraunitary
matrices formed.

The concept of a pseudo-paraunitary matrix is introduced in Section
\ref{pseudo} and
construction methods for these are given.  These may also be considered as FIR (finite
impulse response) systems.

\section{Further Notation}

The book \cite{seh} is an excellent reference for background
 material on the algebraic structures used. 

Now $F$ denotes a general field, $R$ denotes a general ring, 
$\cc$  denotes the complex numbers, $\mathbb{R}$ denotes the real
numbers and $\mathbb{Q}$ denotes the rational numbers. Also $\mathbb{F}_q$ 
denotes the finite field of $q$ elements, $R_{n\ti m}$ denotes the
set of $n\ti m$ matrices with coefficients from $R$ and  $R[\bf z]$
denote the polynomial ring with coefficients from $R$ in  commuting
variables ${\bf z}=(z_1,z_2,\ldots, z_s)$. Note that $R[{\bf z}]_{n\ti
m} = R_{n\ti m}[\bf z]$.

Let $R$ be a ring with identity $1_R =1$. (In general $1$ will denote
the identity of the system under consideration.) A {\em complete family of
orthogonal idempotents} is a set $\{e_1, e_2, \ldots, e_k\}$ in $R$
such that \\ (i) $e_i \not = 0$ and $e_i^2 = e_i$, $1\leq i\leq k$;\\ (ii) If
$i\not = j$ then $e_ie_j = 0$; \\ (iii) $1 = e_1+e_2 + \ldots + e_k$. 

The idempotent $e_i$ is said to be {\em primitive} if it cannot be
written as $e_i= e_i^{'}+ e_i^{''}$ where $e_i^{'},e_i^{''}$ are idempotents
such that $e_i^{'}\neq 0,e_i^{''} \neq 0$ and $e_i^{'}e_i^{''}=0$. A 
set of idempotents is said to be {\em primitive} if each
idempotent in the set is primitive.

Various methods for constructing 
 complete orthogonal sets of idempotents  are derived below. 
Such sets always exist  in $FG$, the group ring over a field $F$, when
$char F \not | \, |G|$. See \cite{seh} for properties of group rings and 
related definitions. These idempotent sets are
related to the representation theory of $FG$.

A mapping $^*: R\to R$ in which $r\mapsto r^*, (r\in R)$ 
is said to be an {\em involution} on $R$ if and only if (i) $r^{**} =
r, \, \forall r\in R$, (ii) $(a+b)^* =
a^*+b^*, \, \forall a,b \in R$, and 
(iii) $(ab)^* = b^*a^* , \, \forall a,b \in R$.

We shall be particularly interested in the
case where $^*$ denotes complex conjugate transpose in the case of
matrices over $\cc$ and denotes transpose for matrices over other
fields. Such a mapping $^*$ on group rings is also defined below.

An element $r\in R$ is said to be {\em symmetric} (relative to $^*$) if
$r^* = r$ and a set of elements is said to be symmetric if each element
in the set is symmetric.

$Q\otimes R$ denotes the tensor product of the matrices $Q,R$.


As already noted $A^*$ is used to denote the
complex conjugate transpose of a matrix $A$. Suppose $R$ is a
ring with involution $^*$. Then $^*$ may be extended to matrices over
$R$ as follows. Let $M\in R_{n\ti m}$ and define $M*$ to be the matrix
with each entry $u$ of $M$ replaced by $u^*$. Then define
$M^* = {M*}\T$. This matrix $M^*$ has size $m\ti n$. 
Let $A(\bf{z})$ be a matrix with polynomial in variables ${\bf z}$ 
over some ring with involution $^*$.
Define $A({\bf{z}})^*$ to be
$A^*({\bf{z}}^{-1})$. When $A$ is 
used for $A(\bf{z})$ write $A^*$ to mean $A(\bf{z})^*$.  (In other words
consider `complex conjugate transposed' of a variable $ z$ to be
$z^{-1}$; this is consistent with group/group ring considerations.)

Let $R$ be a ring with involution $^*$. For $w({\bf z})\in R[\bf
z]$ define $w({\bf z})^* = w^*({\bf
z^{-1}})$.
Say $w({\bf{z}})$ is a {\em paraunitary
element} in $R[\bf{z}]$ (relative to $^*$) 
if and only if $w({\bf{z}})w^*({\bf{z^{-1}}})=w({\bf{z}})w({\bf{z}})^*= 1$.

Suppose $K=(B_1, B_2, \ldots, B_k)$ and $L=(C_1,C_2, \ldots, C_k)$ are
rows of blocks of a matrix  $P$ where each block is of  the same size. 
Then define the {\em block inner product} of $K$ and $L$, written
$K\cdot L$,   to
be $K\cdot L= B_1C_1^*+ B_2C_2^*+ \ldots + B_kC_k^*$. This is to
include the case when $B_i, C_j$ are polynomial matrices and the $C_j^*$ are
 defined as above.


\section{Paraunitary elements}\label{1D1}


The building methods using complete sets of orthogonal idempotents for
the 1D paraunitary matrices in this section are generalised later
in section \ref{multi} below and following.
The next section
\ref{idems} considers methods for designing such complete sets of orthogonal
idempotents.

\begin{proposition}\label{first}
Let $I=\{e_1, e_2, \ldots, e_k\}$ be a complete orthogonal set of
idempotents in a ring $R$.  
 Define $u(z) = \di\sum_{i=1}^k\pm
  e_iz^{t_i}$. Then $u(z)u(z^{-1})=1$.
\end{proposition}
\begin{proof} Since $\{e_1, e_2, \ldots, e_k\}$ is a complete set  
 of orthogonal idempotents, $u(z)u(z^{-1}) = e_1^2 + e_2^2+ \ldots +
 e_k^2 = e_1 + e_2 + \ldots + e_k = 1$.
\end{proof}
\begin{corollary}\label{second} If $I$ is symmetric then
 $u(z){u}^*(z^{-1}) = 1$.
\end{corollary}

Thus $u(z)$ is a paraunitary element 
when $I$ is a symmetric orthogonal complete set of idempotents. 

It is not necessary to use primitive idempotents.
Note also that if
$S=\{e_1, \ldots, e_k\}$ is a complete set  of orthogonal idempotents then
 $\{e_i,e_j\},  \, i\neq j,$ may be replaced by $\{e_i+e_j\}$
in $S$ and the result is (still) a complete set of orthogonal
idempotents.  This idea may be used to obtain real paraunitary matrices
from (complex) complete orthogonal sets of idempotents in group rings.

We single out the case $R= F_{n\ti n}$ 
for special mention.
\begin{proposition} Let $\{I_1, I_2, \ldots, I_k\}$ be a complete
  symmetric set of
 orthogonal idempotents in the ring $F_{n\ti n}$ of $(n\ti n)$ 
matrices over $F$. 
Then $W(z)= \sum_{i=1}^k \pm I_iz^{t_i}$ is a paraunitary 1D $n\ti n$
 matrix over $F$ where the $t_i$ are non-negative integers.
\end{proposition}

In the group ring case a paraunitary element in $FG$ with $|G| = n$
gives a paraunitary matrix in $F_{n\ti n}$ via the embedding of $FG$
into $F_{n\ti n}$ as given for example in \cite{hur3}.
 
Suppose $\{I_1,I_2, \ldots , I_k\}$ is an orthogonal symmetric 
complete set of
idempotents in $F_{n\ti n}$ and that $P$ is a unitary matrix. 
Then also $\{P^*I_1P, P^*I_2P, \ldots,
P^*I_kP\}$ is a symmetric complete orthogonal set of idempotents in
$F_{n\ti n}$.

 
 For our purposes say a paraunitary matrix $P$ is separable
if it can be written in the form $P=QR$ or $P=Q\otimes R$ where $Q,R$
are paraunitary with $Q\neq 1, R\neq 1$; otherwise say $P$ is
non-separable.

The following standard lemma is included for completeness and is not needed
subsequently; the proof is omitted. 
\begin{lemma}\label{trans}
Suppose $A(\bf z)$ is a paraunitary matrix. Then $A^*(\bf z)$ and
 $A\T(\bf z)$ are paraunitary matrices. \end{lemma}

\subsection{Modulus 1}

In Proposition \ref{first} the coefficients of the idempotents are 
$\pm 1$ times monomials. This can be extended in $\cc$ to  coefficients with 
modulus $1$ times monomials. In $\mathbb{R}$ and fields of finite characteristic define
$a^*= a$ and then $\pm 1$ are the only elements which satisfy $aa^*=
a^2=1$.

Suppose $\{E_1, E_2. \ldots, E_k\}$ is a complete symmetric 
orthogonal set of idempotents in $F_{n\ti n}$.
Define  $W(z)=\al_1E_1z^{t_1}+
\al_2E_2z^{t_2}+ \ldots +\al_kE_kz^{t_k}$ and then ${W}^*(z^{-1}) =
{\al_1}^*E_1z^{-t_1} + {\al_2}^*E_2z^{-t_2}+\ldots +
{\al_k}^*E_kz^{-t_k}$. Here if $a\in \cc$, then $a^*= \ov{a}$, 
the complex conjugate of $a$, and for other fields $a^*=a$. Use $|a|^2$
to mean $aa^*$ for any field.

Therefore  $W(z){W}^*(z^{-1})= W(z)W(z)^*= |\al_1|^2E_1+|\al_2|^2E_2+ \ldots +
|\al_k|^2E_k$ \, \, \, (**). 
\begin{proposition}\label{gent} $W(z)$ is a paraunitary matrix if and only
 if $|\al_i|^2=1$ for each $i$.
\end{proposition} 
\begin{proof} If each $|\al_i|^2= 1$ then from (**) $W(z){W}^*(z^{-1})
 =1$. If on the other hand $W(z){W}^*(z^{-1})= 1$ then multiplying (**) through
 (on right) by $E_i$ gives $|\al_i|^2E_i= E_i$ from which it follows that
 $|\al_i|^2=1$.
\end{proof} 

Thus Proposition \ref{first} may be generalised as follows:
\begin{proposition}\label{second2} Let $\{E_1, E_2, \ldots, E_k\}$ be a complete
 symmetric  orthogonal set of
idempotents and  $W(z)=\al_1E_1z^{t_1}+
\al_2E_2z^{t_2}+ \ldots +\al_kE_kz^{t_k}$, with $t_j\geq 0$ and $|\al_j|^2
= 1$ for each  $j$. Then $W(z)$
is a paraunitary matrix. 
\end{proposition}

Now in $\cc $, $|\al|^2=1 $ if and only if $\al =
e^{i\theta}$ for real $\theta$ with $i=\sqrt{-1}$ and in $\mathbb{R}$,
$|\al|^2 = 1$ if and only if $\al=\pm 1$. In a field of
characteristic $p$, $|\al|^2 = \al^2= 1$ if and only if  $ \al=1$ or
$\al = -1 = p-1$. 

As expected unitary matrices are built from complete symmetric
orthogonal sets of matrices as per Proposition \ref{second2}:
\begin{proposition}\label{7}  $U$ is a unitary $n\ti n$ matrix over $\cc$ if and
 only if $U = \al_1 v_1^*v_1 + \al_2 v_2^*v_1 + \ldots + \al_nv_n^*v_n$
 where $\{v_1, v_2, \ldots, v_n\}$ is an orthonormal basis for $\cc_n$
 and $\al_i\in \cc, \,|\al_i|=1, \, \forall i$. Further the $\al_i$ are
 the eigenvalues of $U$.
\end{proposition}
\begin{proof} Suppose $U= \al_1 v_1^*v_1 + \al_2 v_2^*v_1 + \ldots +
 \al_nv_n^*v_n$ with $\{v_1, v_2, \ldots, v_n\}$ an orthonormal
 basis and $|\al_i| =1$. Then $Uv_i^*=\al_iv_i^*$ and so the $\al_i$ are
 the eigenvalues of $U$. It follows from Proposition
 \ref{second2} that $U$ is unitary since $\{v_1^*v_1, v_2^*v_2,
 \ldots, v_n^*v_n\}$ is a complete symmetric orthogonal set of
 idempotents. 

Suppose then $U$ is a unitary
 matrix.  It is known that there exists a unitary matrix $P$ such that 
$U=P^*DP$ where $D$ is diagonal with entries of modulus $1$. Then $P=
\begin{ssmatrix} v_1 \\ v_2 \\ \vdots \\ v_n \end{ssmatrix}$ where $\{v_1,
 v_2, \ldots , v_n\}$ is an orthonormal basis (of row vectors) for $\cc_n$
 and $D=\diag(\al_1,\al_2, \ldots, \al_n)$ with $|\al_i| =1$ and the
 $\al_i$ are the eigenvalues of $U$.
Then \begin{eqnarray*}  U = P^*DP  \\ =& (v_1^*, v_2^*, \ldots,
 v_n^*)\begin{ssmatrix} \al_1 & 0 & \ldots & 0\\ 0& \al_2 & \ldots & 0
\\ \vdots & \vdots & \vdots & \vdots \\ 0 & 0 & \ldots & \al_n \end{ssmatrix} 
 \begin{ssmatrix} v_1 \\ v_2 \\ \vdots \\ v_n \end{ssmatrix} \\ =& 
(\al_1v_1^*, \al_2v_2^*, \ldots, \al_nv_n^*)\begin{ssmatrix} v_1 \\ v_2
					    \\ \vdots \\ v_n
					    \end{ssmatrix}
\\  = &\al_1v_1^*v_1 + \al_2v_2^*v_2 + \ldots + \al_nv_n^*v_n.
\end{eqnarray*}
 
\end{proof}

Thus  unitary matrices are generated by complete symmetric orthogonal
sets of idempotents formed from the diagonalising unitary matrix. Notice
that the $\al_i$ are the eigenvalues of $U$.

For example consider the real orthogonal/unitary matrix
$U=\begin{ssmatrix} \cos \theta & \sin \theta \\ -\sin \theta & \cos
    \theta \end{ssmatrix}$.
This has eigenvalues $e^{i\theta}, e^{-i\theta}$ and
$P=\frac{1}{\sqrt{2}}\begin{ssmatrix} -1 &-i \\ i& 1 \end{ssmatrix}$ is
a diagonalising unitary matrix. Take the rows $v_1=
\frac{1}{\sqrt{2}}(-1,-i), \, v_2=\frac{1}{\sqrt{2}}(i,1)$
of $P$ and consider the complete orthogonal symmetric set of idempotents
$\{P_1 = v_1^*v_1 = \frac{1}{2}\begin{ssmatrix}1 & -i \\ i & 1
\end{ssmatrix}, P_2 = v_2^*v_2 = \frac{1}{2}\begin{ssmatrix}1 & i \\ -i & 1
\end{ssmatrix} \}$.

Then applying Propositon \ref{7} gives  $U=e^{i\theta}P_1+ e^{-i\theta}P_2 =
\frac{1}{2}e^{i\theta}\begin{ssmatrix}1 & -i \\ i & 1 
		      \end{ssmatrix} + \frac{1}{2}e^{-i\theta}
\begin{ssmatrix}1 & i \\ -i & 1
\end{ssmatrix} $, which may be checked independently.



\subsection{Products}

A product of paraunitary matrices and the tensor
product of paraunitary matrices are also 
paraunitary matrices. Thus further 
paraunitary matrices may be designed using these products from those
already constructed . 

\section{Complete orthogonal sets of idempotents}\label{idems}
Paraunitary matrices are designed from complete
symmetric sets of orthogonal idempotents in  section \ref{1D1}
and also in later sections. 
Here we  concentrate on how such sets may be constructed.


\subsection{Systems from orthonormal bases}\label{joint}

Let $V= F^n$. Assume $F^n$ has an inner product so that the notion of
orthonormal basis exists in $V$ and its subspaces. In $\R^n$ and
$\cc^n$ the 
inner product is $vu^*$ for row vectors $v,u$ where $^*$ denotes complex
conjugate transpose; in $\R^n$, $w^* = w\T$, the transpose of $w$.

Suppose now $V= V_1 \oplus V_2 \oplus \ldots \oplus V_k$ is any
direct decomposition of $V$. Let $P_i$ denote the projection of $V$ to
$V_i$. Then $P_i$ is a linear transformation on $V$ and (i) $1= P_1+P_2+
\ldots +P_k$; (ii) $P_i^2= P_i$; (iii) $P_iP_j= 0, i\neq j$.

Thus $\{P_1, P_2, \ldots, P_k\}$ is complete orthogonal set of idempotents.
If each $P_i$ is an orthogonal projection then this set is a complete
symmetric orthogonal set of idempotents. 

 The matrix of $P_i$ may be obtained as follows when $P_i$ is an
 orthogonal projection. Let $\{w_1, w_2, \ldots,
w_s\}$ be an orthonormal basis for $V_i$ and consider $w\in V$. Then $w=
v_i+\hat{w}$ where $\hat{w}\in  V_1 \oplus V_2 \oplus \ldots 
\hat{V_i}\ldots \oplus V_k$ and $v_i\in V_i$. Here $\hat{V_i}$ means
omitting that term. Then $P_i: V \to V_i$ is given by $w\mapsto v_i$.
Now $v_i = \al_1w_1+ \al_2w_2 + \ldots + \al_kw_k$. Take the inner
product with $w_j$ to get $\al_j = v_iw_j^* = ww_j^*$.
Hence $P_i: w\mapsto w(w_1^*w_1 + w_2^*w_2+ \ldots + w_s^*w_s)$. Thus
the matrix of $P_i$ is $w_1^*w_1+ w_2^*w_2+ \ldots + w_s^*w_s$. 

On the other hand suppose $\{P_1,P_2, \ldots, P_k\}$ is a complete
symmetric orthogonal set of idempotents in $F_{n\ti n}$. Then $P_i$ defines a
linear map $V\to V$ by $P_i:v\mapsto vP_i$. Let $V_i$ denote the image
of $P_i$. Then it is easy to check that $V=V_1\oplus V_2 \oplus \ldots
\oplus V_k$.  

The case when each $V_i$ has dimension $1$ is worth looking at separately. 
Suppose $\{o_1, o_2, \ldots, o_n\}$ is an orthonormal basis for $F^n$. 
 Such bases come up naturally in unitary matrices. Let $P_i$ denote the
 projection of $F^n$ to the space generated by $o_i$. Then $P=\{P_1, P_2,
 \ldots, P_n\}$ is an orthogonal symmetric complete set of idempotents in the
 space of linear transformations of $F^n$. It is easy to obtain the
 matrices of $P_i$. The
 matrices $P_i$ may be combined and the resulting set is (still)  a 
complete symmetric 
 orthogonal sets  of idempotents. For example $(P_i + P_j)$ ($i\neq j$)
is still idempotent and is the projection of $F^n$ to the space generated by
 $\{o_i,o_j\}$; replace $\{P_i, P_j\}$ by $(P_i+P_j)$ in $P$ and the new
 set is (still) an orthogonal symmetric complete set of idempotents.
Then $\rank(P_i+P_j) = \rank(P_i)+ \rank(P_j)$ also --  see Lemma
\ref{trrank} below.

For example  $\{v_1=\frac{1}{3}(2,1,2), v_2= \frac{1}{3}(1,2,-2), v_3=
\frac{1}{3}(2,-2,-1)\}$ is an
orthonormal basis for $\R^3$. The projection matrices are respectively 
$P_1 = v_1\T v_1 = \frac{1}{9}\begin{ssmatrix} 4 & 2 & 4 \\ 2&1&2 \\ 4 &2
			      &4 \end{ssmatrix},  P_2= v_2\T v_2 =
			      \frac{1}{9}
\begin{ssmatrix} 1 & 2 & -2 \\ 2&4&-4 \\
			       -2 &-4 & 4\end{ssmatrix}, 
P_3 = v_3\T v_3 = \frac{1}{9}\begin{ssmatrix}4 & -4&-2 \\ -4 &4 & 2
			      \\ -2 & 2& 1\end{ssmatrix}$.

Thus $\{P_1, P_2, P_3\}$ is a complete symmetric 
orthogonal set of idempotents and
each $P_i$ has rank 1.
Set  $\hat{P}_2 = P_2+P_3$ and then $\{P_1,\hat{P}_2\}$
is a complete symmetric orthogonal set of idempotents also and
$\rank(\hat{P}_2 ) = 2$.   

Note that the inner
product in $\cc^n$ is $vu^*$ for  row vectors $v,u$ where $^*$ denotes
{\em complex 
conjugate} transposed. For example $\{\frac{1}{\sqrt{2}}(-i, 1),
\frac{1}{\sqrt{2}}(i,1)\}$ is an orthonormal  
basis for $\cc^2$. Projecting then gives the
complete orthogonal  symmetric set of idempotents $\{P_1 =
\frac{1}{2}\begin{ssmatrix} 1 &-i \\ i & 1 \end{ssmatrix}, P_2 =
\frac{1}{2}\begin{ssmatrix}1 &i \\ -i & 1 \end{ssmatrix}\}$.

\subsection{Orthogonal idempotents systems from unitary/paraunitary}\label{jert}
Let $U$ be a unitary or paraunitary $n\ti n$ matrix in variables $\bf
z$ say over $R$. Then the rows
$\{v_1,v_2,\ldots, v_n\} $ of $U$ satisfy $v_iv_i^*=1,\, v_iv_j^* = 0, \,
i\neq j$. 

Define $P_i=v_i^*v_i$ for $i=1,2,\ldots, n$. Then $P_i$ is an $n\ti n$
matrix of rank $1$.
\begin{proposition} $\{P_1, P_2,\ldots, P_n\}$ is a complete symmetric
  orthogonal set of idempotents in $R_{n\ti n}[\bf z, z^{-1}]$.
\end{proposition}
\begin{proof} It is easy to check that $P_i^*=P_i, 
\, P_iP_i=P_i, \, P_iP_j=0,  \, i\neq j$. It is necessary to show that the set
is complete.

Let $A=P_1+P_2+\ldots + P_n$. Note that $P_iv_i^* = v_i^*, P_iv_j^*=0,
\, i \neq j$. Then $Av_i^* = v_i^*$.  Thus $A$ has $n$ linearly
independent eigenvectors corresponding to the eigenvalue $1$. Hence
$A=I_n$.
\end{proof}

\subsubsection{Diagonals} 
In $R_{n\ti n}$  
let $E_{ii}$ denote the matrix with $1=1_R$ 
on the (diagonal) $(i,i)$ 
position and $0$ elsewhere. Then $W=\{E_{11}, E_{22}, \ldots,
E_{nn}\}$ is a complete symmetric 
orthogonal set of idempotent matrices. This is a special case of 
section \ref{joint} but is worth mentioning separately; paraunitary 
matrices have been  designed from $W$ which, although not
generally useful in themselves directly, may  be combined
with other designed paraunitary matrices with which they do not
commute in general.

\subsection{Group rings}\label{grring}
Group rings are a neat way with which to  obtain complete orthogonal symmetric
sets of idempotents. These systems  have nice structures from which
properties of the paraunitary matrices designed may be deduced. 
Let $w= \sum_{g\in G} \alpha_gg$ be an element in the  group ring
$FG$ and  $W$ denotes the matrix of $w$ as
defined in \cite{hur3} and \cite{hur1}. This matrix $W$ depends on the
listing of the elements of $G$ and relative to this listing $\phi: w\mapsto W$
is an embedding of $FG$ into the ring of $n\ti n$ matrices, $F_{n\ti n}
$, over $F$ where $n=|G|$.
The {\em transpose}, $w\T$, of
$w$ is $w\T = \sum_{g\in G} \alpha_g g^{-1}$. Note that the matrix of
$w\T$ is then $W\T$. 

Over $\cc$ define ${w}^* = \sum_{g\in G} \ov{\al_g}g^{-1}$ where bar
denotes complex conjugate.  Note then that for a group ring element
 $w$ with corresponding matrix $W$ the matrix of $w^*$ is indeed $W^*$.

Say an element $w({\bf{z}})\in FG[{\bf{z}}]$ is a paraunitary group ring
element if and only if $w({\bf{z}})w^*({\bf{z^{-1}}}) =1$ and this
happens if and only if  the corresponding $W({\bf{z}})\in F_{n\ti
n}[{\bf{z}}]$ is a paraunitary matrix (where $n = |G|$). The $W({\bf{z}})$
obtained from $w({\bf{z}})$ in this case is termed a group ring
paraunitary matrix. 

Group rings are a rich source of complete sets of orthogonal
idempotents and group rings have a rich structure within which
properties of the paraunitary matrices so designed may be obtained. 

The theory brings representation theory and  character theory in group
rings into play. The orthogonal
idempotents are obtained from the conjugacy classes and character
tables, see e.g.\ \cite{seh}. The orthogonal sets of idempotents 
 depend on the field under consideration and classes of 
 paraunitary matrices over different fields such as $\mathbb{Q}$,
$\mathbb{R}$ or finite fields  are also obtainable.

The primitive central idempotents of the complex group algebra
$\mathbb{C}G$ are given by $e(\chi)= \frac{\chi(1)}{|G|} \sum_{g\in G}
\chi(g^{-1})g$ where $\chi$ runs through the irreducible (complex)
characters $\chi$ of $G$, see \cite{seh}, Theorem 5.1.11 page 185, where
the $e_i$ are expressed as $e_i= \frac{\chi_i(1)}{|G|} \sum_{g\in
G}\chi_i(g^{-1})g$.
  
The idempotents from group rings are automatically symmetric. 
\begin{theorem}\label{thm:third} For the group idempotents $e_i$,
 $e_i^* = e_i$. 
\end{theorem}
\begin{proof} This is a matter of showing that the coefficients 
  $g$ and $g^{-1}$ in each $e_i$ are complex conjugates of one
  another. But this
  is immediate  as it is well-known that $\chi(g^{-1}) = \overline{\chi({g})}$,
  and thus the result follows from the expression for $e_i$ given above.
\end{proof} 

Let $E_i$ denote the matrix of $e_i$ as per an embedding of the group
ring into the ring of matrices as for example in \cite{hur3}. 
\begin{corollary}\label{idem3}
Let $\{e_1,e_2, \ldots,e_k\}$ be a complete set of orthogonal idempotents in
 a group ring and define $U(z)=\di\sum_{i=1}^{k}\pm E_iz^{t_i}$ where
 the $t_i$ are non-negative integers.  
 Then $U(z)$ is a paraunitary matrix.
\end{corollary}
\begin{corollary}\label{idem4}
Let $\{e_1,e_2, \ldots,e_k\}$ be a complete set of orthogonal idempotents in
 a group ring over $\cc$ and define $U(z)=\di\sum_{i=1}^{k}\al_i
 E_iz^{t_i}$ where the $t_i$ are non-negative integers and 
 $|\al_i|=1 $. 
 Then $U(z)$ is a paraunitary matrix.
\end{corollary}
 
The formula for the $\{e_i\}$ as given above  (taken from \cite{seh}) may be
used to construct complete orthogonal sets of idempotents. 
The Computer Algebra packages GAP and Magma can construct character tables
and conjugacy classes from which complete sets of orthogonal
idempotents in group rings may be obtained. The literature contains 
other numerous methods for finding complete (symmetric) orthogonal 
sets of idempotents in group rings. 

In general the
paraunitary matrices designed using orthogonal sets of idempotents  in
the group ring over $\cc$ have complex coefficients 
but specialising and combining idempotents allows the 
design so that the coefficients may be in  $\mathbb{R}$, the real numbers, or in 
$\mathbb{Q}$, the rational numbers. When the group ring of the 
symmetric group  $S_n$ is used the paraunitary matrices derived by these
methods all have
coefficients automatically in $\mathbb{Q}$ and when the group ring of 
dihedral group $D_{2n}$ is used the coefficients are in $\mathbb{R}$. In general
idempotents occur in complex conjugate pairs  
and these  may be combined to give real coefficients resulting in
paraunitary matrices with real coefficients.

Most of the results hold in the case when the characteristic of $F$ does not
divide the order of $G$; in this case 
this  means that the characteristic of $F$ does not
divide the size $n$ of the $(n\times n)$ matrices under consideration. In
these cases also it may be necessary to extend the field to include
roots of certain polynomials.

\subsection{Tensor products} 
It is easy to check that the tensor product of paraunitary matrices is
also a paraunitary matrix. If $P=QR, S=TV$ then $P\otimes S = (QT)\otimes
(RS)$ when the products $QT,RS$ can be formed.

Complete orthogonal sets of idempotents may be
designed using products of these sets. Suppose $\{e_0,e_2, \ldots, e_k\}$ is a
complete orthogonal set of idempotents in $F_{n\ti n}$ and $\{f_0,f_1,
\ldots, f_s \}$ is a complete orthogonal set of idempotents in $F_{k\ti
k}$.
Then $\{e_i\otimes  f_j \, | \, 0\leq i \leq k, 1\leq j \leq s\}$ is a complete
orthogonal set of idempotents in $F_{nk \ti nk}$. Here $\otimes$ denotes
tensor product. If both $\{e_0,e_2, \ldots, e_k\}$ and $\{f_0,f_1,
\ldots, f_s \}$ are symmetric then so is the resulting tensor product
set. The details are omitted.

If $\{ e_i \, | \, 1\leq i \leq k\}$ and $\{ f_j \, | \, 1\leq j \leq
s \}$ are 
complete orthogonal sets of idempotents within group rings 
$FG, FH$ respectively 
 then $\{e_if_j \, | \, 1\leq i \leq k, 1\leq j \leq s\}$ is 
a complete orthogonal set of matrices in $F(G\cross H)$.
Suppose $e_i \mapsto E_i, f_j \mapsto F_j$ gives an embedding into
matrices, then 
$e_if_j \mapsto E_i\otimes F_j$ gives an embedding into $F(G\cross H)$;
 this may be deduced from \cite{hur3} and details are
omitted. 


\subsection{Examples of paraunitary matrix from orthonormal bases}
\begin{enumerate} \item  
The complete orthogonal symmetric systems of idempotents 
$P_1 =\frac{1}{9}\begin{ssmatrix} 4 & 2 & 4 \\ 2&1&2 \\ 4 &2
			      &4 \end{ssmatrix},  P_2=
			      \frac{1}{9}
\begin{ssmatrix} 1 & 2 & -2 \\ 2&4&-4 \\
			       -2 &-4 & 4\end{ssmatrix}, 
P_3 =  \frac{1}{9}\begin{ssmatrix}4 & -4&-2 \\ -4 &4 & 2
			      \\ -2 & 2& 1\end{ssmatrix}$
were obtained in section \ref{joint}. Then $W(z)= P_1z^2+P_2z+P_3z^3$ is a
paraunitary matrix. 
\item Let $z=e^{i\theta}$ in $W$ in 1.\ gives a unitary matrix, $T$
  say. The rows of $T$ form an orthonormal basis for $\cc^3$. These
  rows may then be used to form a complete symmetric
  orthogonal set of idempotents from which paraunitary matrices may be
  constructed. This process could be continued.
\item In $\cc C_3$, where $C_3$ is the cyclic group of order $3$, 
the orthogonal complete set of
idempotents formed are $Q_1 = \frac{1}{3}\begin{ssmatrix}
						      1&1 &1 \\ 1&1 
&1 \\ 1&1 &1 \end{ssmatrix}, Q_2 =\frac{1}{3}\begin{ssmatrix} 1& \om &
					      \om^2 \\ \om^w &
1 & \om \\ \om & \om^2 & 1 \end{ssmatrix}, 
 Q_3 = \frac{1}{3}\begin{ssmatrix} 1 & \om ^2 &\om \\ \om & 1 & \om^2 \\
		   \om^2&\om & 1 \end{ssmatrix}$ where $\om$ is a
		   primitive 3rd root of unity. 

Then $Q(z) = Q_1+Q_2z^3+Q_3z^2$ is a paraunitary matrix.
\item Give values of modulus $1$ to $z$ in $Q(z)$ above  and get a
  unitary matrix $R$. Use the rows of $R$ to form a further complete
  symmetric set of idempotents from which paraunitary matrices may be formed.
\item  Combine $Q(z)$ in 3.\ 
with $W(z)$ in 1.\ to give for example the paraunitary matrix
       $Q(z)W(z)Q(z)$.    
 
\end{enumerate}
We give some examples from orthogonal sets
of idempotents derived from group rings. The group ring idea 
 is used later as a prototype in which to extend the method for the
 design of non-separable paraunitary matrices. The complete
 orthogonal sets of idempotents obtained from group rings are
 automatically symmetric as noted in  \thmref{third}.

Recall that  $\cir(a_0,a_1,\ldots, a_{n-1})$ denotes 
the circulant $n\ti n$ matrix with
first row $(a_0,a_1,\ldots, a_{n-1})$.

Consider $\cc C_n$  where $C_n$ is a cyclic group of order $n$.

\begin{enumerate}
\item 
 When $n=2$ the (primitive) orthogonal set of
 idempotents consists of $\{e_0= 1/2(1+g), e_1 = 1/2(1-g)\}$, where 
 $g$ generates $C_2$. Thus paraunitary matrices may be formed from 
$E_0=\frac{1}{2}\left(\begin{array}{cc}1&1\\ 1&1 \end{array}\right)$ and
$E_1=\frac{1}{2}\left(\begin{array}{cc}1&-1\\ -1&1 \end{array}\right)$
giving for example $\frac{1}{2}\left(\begin{array}{cc}1+z&1-z\\ 1-z&1+z
\end{array}\right)$.  (Looks familiar?) 
 \item  These may be combined with paraunitary matrices formed from $E_{11} =
\begin{pmatrix}1 &0 \\ 0&0\end{pmatrix}, E_{22}=
\begin{pmatrix}0&0\\ 0&1\end{pmatrix}$.
Note that $E_{11},E_{22}$ do not commute with $E_0,E_1$.
For example the following is a  paraunitary matrix:

$$\begin{pmatrix}1&0\\
       0&z\end{pmatrix}\frac{1}{2}\begin{pmatrix}z+z^2
&z-z^2\\ z-z^2& z+z^2\end{pmatrix}\begin{pmatrix}z^2&0 \\
						0&z^3
					       \end{pmatrix}
\frac{1}{2}\begin{pmatrix}z^2+z^3&z^2-z^3\\
		 z^2-z^3&z^2+z^3\end{pmatrix}$$

They may also be combined with paraunitary matrices formed from
orthonormal bases as in section \ref{joint} such as $\{P_1 =
\frac{1}{2}\begin{ssmatrix} 1 &-i \\ i & 1 \end{ssmatrix}, P_2 =
\frac{1}{2}\begin{ssmatrix}1 &i \\ -i & 1 \end{ssmatrix}\}$.

The determinant of these matrices which are powers of $z$ may be obtained
from Theorem \ref{det} below.   
\item 
The primitive orthogonal idempotents for a cyclic group are related to
the Fourier Matrix.

\item In $\cc C_4$, for example, the orthogonal primitive idempotents are 
$e_1 = \frac{1}{4}(1+a + a^2 +a^3),
e_2 =\frac{1}{4}( 1 + \om a+\om^2 a^2+\om^3 a^3), e_3 = \frac{1}{4}(1
- a + a^2  -a^3),
e_4 = \frac{1}{4}(1+\om^3a+\om^2 a^2 + \om a^3)$ from which $4\ti 4$
paraunitary matrices may be constructed. Here
$\om$ is a primitive $4^{th}$ root of unity and in this case  $\om^2 = -1$.  

Notice that $e_i = \ov{e}_i\T$ as could be deduced from  \thmref{third}.

\item Combine the $e_i$ to get real sets of orthogonal
idempotents. Note that it is simply enough to combine the conjugacy
classes of $g$ and $g^{-1}$. In this case then we get 

$\hat{e}_1 = e_1 = \frac{1}{4}(1+a+a^2+a^3), \hat{e}_2 = e_2+e_4 = 
\frac{1}{2}(1 -a^2), \hat{e}_3 = e_3 = \frac{1}{4}(1-a+a^2-a^3)$, which can
then be used to construct real paraunitary $4\ti 4$ matrices.

\item Using $C_2\ti C_2$ gives different paraunitary matrices. Here the set
of primitive orthogonal idempotents consists of 
 $f_1 = \frac{1}{4}(1+a+b+ab), f_2 = \frac{1}{4}(1- a +b-ab), f_3=
\frac{1}{4}(1-a-b+ab), f_4 = \frac{1}{4}(1+a- b-ab)$
and the paraunitary matrices derived are all real. 

\item 

The paraunitary matrices produced from $C_4$ from $C_2\ti C_2$ and from
$E_{11}, E_{22}, E_{33}, E_{44}$ may then be combined to produce further
($4\ti 4$) paraunitary matrices. So for example the following
$4\ti 4$ is a paraunitary matrix:
$$
(E_1+E_2z+E_3z^3+E_4z^2)(E_{11}+E_{22}z+E_{33}z^3+E_{44}z^2)(F_1z+F_2z^2+F_3z^3
+ F_4z^2)$$
Again the determinant of the matrix may be obtained from Theorem
\ref{det}. 

The $E_i,F_j$ are derived from the $e_i,f_j$ (as per \cite{hur3}) 
so for example 
$E_2 =\frac{1}{4}\cir(1,\om,\om^2,\om^3)$, $F_3=
\frac{1}{4}\left(\begin{smallmatrix}1&-1&-1&1\\ -1&1&1&-1\\ -1&1&1&-1\\ 1&-1&-1&1
\end{smallmatrix}\right) $.
\end{enumerate}

\subsection{Get real} By combining complex conjugate idempotents 
in a complete orthogonal sets of complex idempotents,  
 real paraunitary matrices may be obtained. We illustrate this with an example.

Suppose $\{e_0,e_1,e_2, e_3,e_4,e_5\}$ is the complete 
set of  primitive idempotents in $\cc C_6$. Here then $e_i=\frac{1}{6}(1
+\om^ig+\om^{2i}g^2 + \om^{3i}g^3+\om^{4i}g^4+ \om^{5i}g^5)$ where $\om
= e^{2i\pi/6 }$ is a
primitive $6^{th}$ root of unity and $C_6$ is generated by $g$.
 
Then $\ov{e_0}=e_0, \ov{e_1}=e_5, \ov{e_2}=e_4, \ov{e_3}= e_3$.

Let $\theta = 2\pi/6$. Note that $\cos(\theta)=\cos(5\theta),
\cos(2\theta)= \cos(4\theta)$. 
Now combine $e_1$ with $e_5$ and $e_2$ with $e_4$ to get 
$\hat{e_1} = \frac{2}{6}(1 + \cos (\theta)g+ \cos (2\theta)g^2+
\cos(3\theta)g^3+ \cos(4\theta)g^4+ \cos(5\theta)g^5)$ and $\hat{e_2}=
\frac{2}{6}(1+ \cos(2\theta)g+\cos(2\theta)g^2)+\cos(2\theta)g^3+
\cos(2\theta)g^4+ \cos(2\theta)g^5$. This gives the real orthogonal
complete set of idempotents $\{e_0,\hat{e_1},\hat{e_2},e_3\}$ from which
real paraunitary matrices may be constructed. The ranks of the
idempotents and determinants of the paraunitary matrices formed may be
deduced from Lemma \ref{trrank} and Theorem \ref{det}.
 
\subsection{Symmetric, dihedral  groups}{\label{symm6}} 


Let $D_{2n}$ denote the dihedral group of order $2n$. As every element
in $D_{2n}$ is conjugate to its inverse, the complex characters of $D_{2n}$ are
real. Thus the paraunitary matrices obtained directly from the complete
orthogonal set of idempotents in  $\cc D_{2n}$
have real coefficients. The characters $D_{2n}$ 
are contained in an
extension of  $\mathbb{Q}$ of degree $\phi(n)/2$ and  this is 
$\mathbb{Q}$ only for $2n\leq 6$.   

Let $S_n$ denote the symmetric group of order $n$. Representations and
orthogonal  
idempotents of the symmetric group are known; see for example
\cite{curtis}. The characters of
$S_n$ are rational and thus the paraunitary matrices produced directly
from the complete orthogonal set of idempotents in $\cc S_n$ have rational coefficients.

The paraunitary matrices formed from different group rings (with same
size group) may be combined to form further paraunitary matrices; these
in general will not commute.  

We present an example here from $S_3$, the symmetric group on 3
letters. (Note that $S_3=D_6$.)

Now $S_3 = \{1, (1,2), (1,3), (2,3), (1,2,3), (1,3,2)\}$ where these are
cycles. We also use this listing of $S_3$ when constructing matrices. 

There are three conjugacy classes:
$K_1 = \{1\}$; $K_2 = \{(1,2), (1,3)\},  (2,3)$; $K_3 =\{
(1,2,3),(1,3,2)\}$.

Define \\ $\hat{e}_1 = 1 + (1,2)+(1,3)+(2,3)+(1,2,3)+(1,3,2)$,
\\ $\hat{e}_2 = 1 - \{ (1,2) + (1,3)+(2,3)\} + (1,2,3)+(1,3,2)$,\\
$\hat{e}_3 = 2 - \{(1,2,3) + (1,3,2)\}$,

and $e_1 = \frac{1}{6}\hat{e}_1; e_2 = \frac{1}{6}\hat{e}_2; e_3 =
\frac{1}{3}\hat{e}_3$. Then $\{e_1,e_2,e_3\}$ form a complete
  orthogonal set of  idempotents and may be used to construct
  paraunitary matrices. 

The $G$-matrix of $S_3$ (see \cite{hur3}) is 
$\left(\begin{smallmatrix}1 & (12)&(13)&(23)&(123)&(132)\\
  (12)&1&(132)&(123)&(23)&(13) \\ (13) &(123) &1 & (132) &(12)&(23) \\
  (23) &(132)&(123)&1&(13)&(12) \\ (132)& (23)&(12)&(13)&1&(123) \\
  (123)&(13)&(23)&(21)&(132)&1\end{smallmatrix}\right).$

Thus the matrices of $e_1,e_2,e_3$ are respectively 

$$E_1=\frac{1}{6} \left(\begin{smallmatrix} 1 & 1&1&1&1&1\\
  1&1&1&1&1&1 \\ 1&1&1&1&1&1 \\ 1&1&1&1&1&1 \\ 1&1&1&1&1&1
  \\ 1& 1&1&1&1&1 \end{smallmatrix}\right),
E_2=\frac{1}{6} \left(\begin{smallmatrix}1 & -1&-1&-1&1&1\\
  -1&1&1&1&-1&-1 \\ -1&1&1&1&-1&-1 \\ -1&1&1&1&-1&-1 \\ 1&-1&-1&-1&1&1
  \\ 1& -1&-1&-1&1&1 \end{smallmatrix}\right),
E_3=\frac{1}{3} \left(\begin{smallmatrix} 2 & 0&0&0&-1&-1\\
  0&2&-1&-1&0&0 \\ 0&-1&2&-1&0&0 \\ 0&-1&-1&2&0&0 \\ -1&0&0&0&2&-1
  \\ -1& 0&0&0&-1&2 \end{smallmatrix}\right).$$

Note that $E_1,E_2$ have $\rank 1$ and that $E_3$ has $\rank 4$; 
the proof for the  ranks of these $E_i$ in general  
is contained in Lemma  \ref{trrank}.

Thus for example the following are paraunitary matrices:

$E_2+E_1z+E_3z^2$, $E_3+(E_1+E_2)z$, $E_1+E_2 + E_3z^2$.

The paraunitary matrices formed from these idempotent matrices may then
be combined with paraunitary matrices formed using 
complete orthogonal idempotents
obtained from $\cc C_6$ and ones using $\{E_{11}, E_{22}, E_{33}, E_{44},E_{55},
E_{66}\}$.  

Let $\{f_1, \ldots, f_6\}$ be the orthogonal idempotents from $\cc C_6$.
Let $w_1=\sum_{i=1}^3E_iz^{t_i},
w_2=\sum_{i=1}^6E_{ii}z^{k_i}, w_3=\sum_{i=1}^6F_iz^{l_i}$. ($F_i$ is
the matrix of $f_i$.) Then
products of $w_1,w_2,w_3$ are paraunitary matrices. Note that the $w_i$
do not commute.

\subsection{Finite fields}\label{finite}
Here we consider  constructing examples of complete symmetric 
sets of orthogonal
idempotents over finite fields.

Suppose $\{v_1,v_2,\ldots, v_k\}$ is a orthogonal basis  for $F^k$ under
$(u,v) = uv\T$. Thus $v_iv_j\T = 0$ for $i\neq j$.

Suppose also $(v_i,v_i)= t_i \neq 0$.
Define $P_i=t_i^{-1}v_i\T v_i$ which is a $k\ti k$ matrix. Then 
$P_iP_i\T = t_i^{-1}v_i\T v_i t_i^{-1}v_i\T v_i = t_i^{-1}v\T v_i = P_i$ and 
$P_iP_j\T= t_i^{-1}v_i\T v_i t_j^{-1}v_j\T v_j = 0$ for $i\neq j$.

It also follows that $\di\sum_{j=1}^kP_j= 1$.
To see this consider $A= P_1+P_2+\ldots + P_k$. Then $v_iA = v_iP_i= v_i$ as 
$v_iP_j=0$ for $i\neq j$ and $v_iP_i= v_it^{-1}v_i\T v_i=v_i$. Hence
$v_iA =v_i$.
Let $Q=\begin{pmatrix}v_1 \\ v_2 \\ \vdots \\
	    v_k\end{pmatrix}$. Then $Q$ is non-singular as
	    $\{v_1,v_2,\ldots, v_k\}$ is linearly independent.
Then $QA=Q$ and hence $A=1$. 

Note that in the above we do not need to take the square root of 
elements. 

Another way could be to construct such sets over $\mathbb{Q}$ and when the
denominators do not involve a prime dividing the order of the field it
is then possible to derive complete symmetric orthogonal
sets of idempotents over the finite field.

For example the complete orthogonal symmetric systems of idempotents 
$P_1 =\frac{1}{9}\begin{ssmatrix} 4 & 2 & 4 \\ 2&1&2 \\ 4 &2
			      &4 \end{ssmatrix},  P_2=
			      \frac{1}{9}
\begin{ssmatrix} 1 & 2 & -2 \\ 2&4&-4 \\
			       -2 &-4 & 4\end{ssmatrix}, 
P_3 =  \frac{1}{9}\begin{ssmatrix}4 & -4&-2 \\ -4 &4 & 2
			      \\ -2 & 2& 1\end{ssmatrix}$
were obtained in section \ref{joint}. 

Over a field of characteristic $2$ these come to the trivial set
$\{\begin{ssmatrix} 0&0&0\\ 0&1& 0\\ 0&0&0
\end{ssmatrix},\begin{ssmatrix} 1&0&0\\ 0&0& 0\\ 0&0&0
\end{ssmatrix},\begin{ssmatrix} 0&0&0\\ 0&0& 0\\ 0&0&1 \end{ssmatrix}\}$
of symmetric complete orthogonal set of idempotents.

Over the field $\mathbb{F}_5$ of 5 elements they become (note that here 
$9^{-1} =4$):
$\{\begin{ssmatrix} 1&3&1 \\ 3&4&3 \\ 1& 3&1 \end{ssmatrix},
\begin{ssmatrix} 4&3&2 \\ 3&1& 4 \\ 2&4& 1 \end{ssmatrix},
\begin{ssmatrix} 1 & 4&2 \\ 4&1 & 3 \\ 2 & 3 &4\end{ssmatrix}\}$.

This is a complete symmetric
orthogonal set of idempotents in $\mathbb{F}_5$ which may be checked
independently. 


The following are complete symmetric orthogonal
sets of idempotents over $\mathbb{F}_7$:

 $\{\begin{ssmatrix} 2&1&2\\ 1&4&1 \\ 2&
1&2 
\end{ssmatrix},
\begin{ssmatrix} 4&1&6 \\ 1&2& 5 \\ 6&5& 2 \end{ssmatrix},
\begin{ssmatrix} 2 & 5&6 \\ 5&2& 1 \\ 6 & 1 &4\end{ssmatrix}\}$,
$\{\begin{ssmatrix} 6&5&6 \\ 5&3&5 \\ 6& 5&6 \end{ssmatrix},
\begin{ssmatrix} 5&2&5 \\ 2&5& 2 \\ 5&2& 5 \end{ssmatrix},
\begin{ssmatrix} 4 & 0&3 \\ 0&0 & 0 \\ 3 & 0 &4\end{ssmatrix}\}$.

These different sets 
 may be used to construct paraunitary matrices over $\mathbb{F}_7$ and in
a later section are used to show how to construct as an example  
a non-separable paraunitary matrix over a 
 finite field. 
 
\subsection{1D building blocks}\label{tub} 

The great factorisation theorem of Belevitch and Vaidyanathan, see
\cite{vaid} (pp. 302-322), is that matrices of the form 
$H(z)= 1 - vv^* + zvv^*$, where $v$ is any unit column vector
($v^*v=1$),  are the
building blocks for 1D paraunitary matrices over $\cc$. 

Consider $F_1= vv^*$ where $v$ is a unit column vector and so 
$v^* v= 1$. Thus $F_1F_1= vv^* vv^*  = vv^* = F_1$ and so $F_1$ is an
idempotent. Hence $\{F_1 = vv^*, F_2 = 1-F_1= 1- vv^*\}$
is a complete symmetric 
orthogonal set of these (two) idempotents with $\rank F_1= 1$
and $\rank F_2=(n-1)$  
where the matrices have size $n\ti n$; see Theorem \ref{trrank} below
for rank result. 
Then $H(z) = F_2+zF_1$ and 
hence the paraunitary 1D matrices are built from complete symmetric 
orthogonal sets of two idempotents, one of which has $\rank 1$ and the
other has $\rank (n-1)$.

\begin{proposition} Let $F$ be a field in which every element has a square
 root. Suppose also an involution $^*$ is defined on the set of
 matrices over $F$. Then $P$ is a symmetric (with respect to $^*$) 
idempotent of $\rank 1$ in $F_{n\ti n}$
if and only if $P= vv^*$ where $v$ is a column
 vector such that $v^* v=1$.  
\end{proposition}

(Note that `symmetric with respect to $^*$' in the case of $\cc$ in
which $^*$ denotes complex conjugate transposed is normally termed `Hermitian'.)

\begin{proof} If $P=vv^*$ with $v^*v = 1$ then $P$ is a symmetric
 idempotent of $\rank 1$. 

Suppose $P$ is a symmetric idempotent of $\rank 1$ in
 $F_{n\ti n}$. Since $P$ has rank 1 each row is a multiple of any
 non-zero row. Suppose the first row is non-zero and that the first
 entry of this row is non-zero. Proofs for other cases are
 similar. Since $P$ is symmetric it has the form 

$P= \begin{ssmatrix} b_1&b_2 & \ldots & b_n \\ b_2^* & b_2b_2^*/b_1 & \ldots
     & b_nb_2^*/b_1 \\ b_3^* & b_2b_3^*/b_1 & \ldots & b_nb_3^*/b_1 \\ \vdots
     & \vdots & \vdots & \vdots \\ b_n^* & b_2b_n^*/b_1 & \ldots &
     b_nb_n^*/b_1 \end{ssmatrix}$

with $b_1^*=b_1$.

Since $P$ is idempotent it follows that $b_1^2 + |b_2|^2 + \ldots + |b_n|^2
 = b_1$.

Let $v=\frac{1}{\sqrt{b_1}}(b_1, b_2, \ldots, b_n)^*$.
Then $v^* v = \frac{1}{b_1}(b_1^2+|b_2|^2 + \ldots + |b_n|^2) = 1$

and $vv^* = \frac{1}{b_1}\begin{ssmatrix}b_1b_1 & b_1b_2 & \ldots &
			  b_1b_n \\ b_2^*b_1 & b_2^*b_2 & \ldots & b_2^*b_n \\
			  \vdots & \vdots & \vdots & \vdots \\ b_n^*b_1 &
			  b_n^*b_2 &  \ldots & b_n^*
			 b_n\end{ssmatrix}= P$.
\end{proof}  

It is necessary that square roots exist in the field. For example $P=
\begin{ssmatrix} 2 & 1 \\ 1 & 2 \end{ssmatrix}$ over  $\mathbb{F}_3$ 
 is a symmetric idempotent matrix of rank $1$ but cannot be
written in the form $vv\T$; however $2$ does not have a square root 
in $\mathbb{F}_3$. Note that $P_1= 1-P = \begin{ssmatrix}2&2 \\ 2&2\end{ssmatrix}$ and that $\{P,P_1\}$ is a complete orthogonal set of idempotents in
					  $(\mathbb{F}_3)_{2\ti 2}$. 

Over $\mathbb{F}_3$ the following complete symmetric sets of
idempotents are the building blocks
for $2\ti 2$ matrices:
$\{\begin{ssmatrix} 2&1 \\ 1&2  \end{ssmatrix}, \begin{ssmatrix} 2&2 \\
						  2 & 2
						 \end{ssmatrix}\}$,
$\{\begin{ssmatrix}1 & 0\\ 0 & 0 \end{ssmatrix}, \begin{ssmatrix}0& 0 \\
						  0&1\end{ssmatrix}\}$.
Thus the paraunitary matrices $2\ti2 $ matrices over 
$\mathbb{F}_3$ are built from these sets using 
Proposition \ref{second2} and products. 
These sets are not of Belevitch and Vaidyanathan form and their result
does not apply here.


The 1D building block result  of Belevitch and Vaidyanathan cannot be extended to multidimensions.

Group ring 1D
paraunitary matrices and other 1D paraunitary matrices over $\cc$ 
constructed here
can in theory then be obtained from this
characterisation. Group rings have 
special features and paraunitary matrices from these 
have nice structures. 

\section{Multidimensional}\label{multi}

\subsection{kD with idempotents}{\label{1D}}
\begin{proposition}\label{there} Let $\{E_1, E_2, \ldots, E_t\}$ be a
  complete symmetric  
orthogonal set of idempotents and define products of
non-negative powers of variables by  $w_i({\bf{z}})=
\pm \prod_{j=1}^kz_j^{t_{i,j}}$ for $i=1,2, \ldots,
t$ where ${\bf{z}}= (z_1,z_2, \ldots, z_k)$ and $t_{i,j}$ are
non-negative integers. 
Define $W({\bf{z}}) = \sum_{i=1}^tw_i({\bf{z}})E_i$.
Then  $W({\bf{z}})$ is a $k$-dimensional paraunitary element.
\end{proposition}
\begin{proof}
Since  $\{E_1, E_2, \ldots, E_t\}$ is a complete symmetric orthogonal set of
 idempotents, $W({\bf{z}})W({\bf{z}})^* = E_1^2+E_2^2+ \ldots + E_t^2 =
 E_1+E_2 +\ldots + E_t = 1$.
\end{proof}
 
As before for $a\in \cc$ define $a^*=\ov{a}$, the complex conjugate of
$a$, and for other fields
define $a^*=a$. Define $|a|^2 = aa^*$ in all cases.
\begin{proposition}\label{there1} Let $\{E_1, E_2, \ldots, E_t\}$ be a
  complete symmetric  
orthogonal set of idempotents and define products of
non-negative powers of the variables by  $w_i({\bf{z}})=
\al_i\prod_{j=1}^kz_j^{t_{i,j}}$ for $i=1,2, \ldots,
t$ where ${\bf{z}}= (z_1,z_2, \ldots, z_k)$ and $t_{i,j}$ are
non-negative integers and $|\al_i|^2=1$. 
Define
 $W({\bf{z}}) = \sum_{i=1}^tw_i({\bf{z}})E_i$ and  then $W({\bf{z}})$
is a paraunitary matrix.  
\end{proposition}

The $W({\bf{z}})$ so formed can be combined using products of matrices, or
tensor products of matrices when appropriate,
 to form further paraunitary matrices.

Such paraunitary matrices formed from Propositions \ref{there} and \ref{there1}
can however be shown to be separable but have uses of their own and will
be used later to form (constituents of) non-separable paraunitary matrices.




\subsection{Examples of 2D paraunitary} 

Recall that if $u$ is a group ring element of $FG$ with $|G| =n$ 
then $U$ (capital letter equivalent) 
denotes the matrix of $u$ under the embedding of $FG$ into the
ring $F_{n\ti n}$ of $n\ti n$ matrices over $F$, see  \cite{hur3}.

Let $\{e_0,e_1,e_2\}$ be a (primitive) orthogonal complete set of
idempotents in $\cc C_3$, and thus $\{E_0,E_1,E_2\}$ is an orthogonal
complete symmetric set of idempotents in $\cc_{3\ti 3}$. Let
$\{f_0,f_1,f_2,f_3\}$ be an orthogonal complete set of idempotents in 
$\cc C_4$.
Define $u(z,y) = (e_0+e_1z+e_2z^2)f_0 + (e_0+e_1z+e_2z^2)f_1y +
(e_0+e_1z+e_2z^2)f_2y^2 + (e_0+e_1z+e_2z^2)f_3y^3$ and let $U(z,y)$ be
obtained from $u(z,y)$ by replacing each $e_i$ by $E_i$ and each $f_i$
by $F_i$.  

Then $U(z,y)$ is a paraunitary matrix. However here $u(z,y)
= (e_0+e_1z+e_2z^2)(f_0+f_1y+f_2y^2+f_3y^3)$ and so $U(z,y)$ is
separable. 

Let $u(z,y) = (e_0+e_1z+ e_2z^2)f_0 +
(e_0+e_1z^2+e_2z)f_1y + (e_0z+e_1+e_2z^2)f_2y^2 +
(e_0z+e_1z^3+e_2z^2)f_3y^3$. Then $U(z,y)$ is a paraunitary matrix.

\section{Matrices of idempotents}\label{gyt}
Paraunitary matrices may also be constructed from matrices with 
blocks of complete orthogonal sets of idempotents. 

Consider the following example.
Let $E_0=\frac{1}{2}\left(\begin{smallmatrix} 1& 1 \\ 1& 1
			  \end{smallmatrix}\right), E_1 =
\frac{1}{2}\left(\begin{smallmatrix} 1& -1 \\ -1& 1 \end{smallmatrix}\right)$.

Define  
 $W = \left(\begin{smallmatrix} xE_0 & yE_1 \\ zE_1 &
		 tE_0\end{smallmatrix}\right)=
		 \frac{1}{2}\left(\begin{matrix}x&x & y& -y \\ x&x &
				     -y & y \\ z&-z&t&t \\
				  -z&z&t&t\end{matrix}\right)$.


Then $WW^*= I_4$ as $\{E_0, E_1\}$ is an orthogonal
symmetric complete set of idempotents.
 However $W$ is separable as  

$W=\begin{pmatrix}xE_0+ E_1 & 0\\ 0 &
	tE_0+E_1\end{pmatrix}\begin{pmatrix}E_0 & yE_1 \\ zE_1 &
			     E_0\end{pmatrix}$. \, \, (**)

and each of the matrices on the right in (**) is separable into 1D paraunitary
	matrices.

Here if we let 
	$x=1=t$ then (**) is a trivial product as  the first matrix on
	the right  is
	then the identity. If further $y=1=t$ this produces the matrix 
$H= \frac{1}{2}\left(\begin{smallmatrix} 1 & 1 & 1 & -1 \\ 1 & 1 & -1& 1 \\ 1&-1
	     & 1 & 1 \\ -1 & 1 & 1 & 1 \end{smallmatrix}\right)$ which is a
	common matrix used in quantum theory as non-separable. (This
	matrix $H$ with fraction omitted is a Hadamard regular matrix.) Thus
	non-separability of a paraunitary matrix is a stronger condition by
	comparison. 

Let $Q_0=\frac{1}{2}\left(\begin{smallmatrix} 1& i \\ -i& 1
			  \end{smallmatrix}\right), Q_1 =
\frac{1}{2}\left(\begin{smallmatrix} 1& -i \\ 1& 1 \end{smallmatrix}\right)$.
$W = \left(\begin{smallmatrix} xQ_0 & yQ_1 \\ zQ_1 &
		 tQ_0\end{smallmatrix}\right)$.

Then $W$ is a paraunitary matrix. Now letting the variables have complex
 values of modulus $1$  gives rise to 
complex Hadamard regular matrices as for example 
$\left(\begin{smallmatrix} 1 & i & 1 & -i \\ -i & 1 & i& 1 \\ 1&-i
	     & 1 & i \\ i & 1 & -i & 1 \end{smallmatrix}\right)$.  

Let $\{E_0,E_1,E_2\}$ be an orthogonal symmetric complete 
set of idempotents in $\mathbb{F}_{3\ti 3}$.

Define $W= \begin{pmatrix}xE_0 & yE_1 & z E_2 \\ pE_2 & qE_0 & r
E_1 \\ sE_1&tE_2&vE_0\end{pmatrix}$.

The variables of $W$ are $x,y,z,p,q,r,s,t,v$ which need not necessarily
be distinct. 

Then $WW^* = I_9$.

For example in section \ref{joint} the following complete set of
			      symmetric idempotents was  
			      constructed  
in $\mathbb{Q}_{3\ti 3}$:

$P_1 = v_1\T v_1 = \frac{1}{9}\begin{ssmatrix} 4 & 2 & 4 \\ 2&1&2 \\ 4 &2
									&4 \end{ssmatrix},  P_2= v_2\T v_2 =
			      \frac{1}{9}
\begin{ssmatrix} 1 & 2 & -2 \\ 2&4&-4 \\
			       -2 &-4 & 4\end{ssmatrix}, 
P_3 = v_3\T v_3 = \frac{1}{9}\begin{ssmatrix}4 & -4&-2 \\ -4 &4 & 2
			      \\ -2 & 2& 1\end{ssmatrix}$.

Then  $W= \begin{pmatrix}xP_1 & yP_2 & z P_3 \\ pP_3 & qP_1 & r
P_2 \\ sP_2&tP_3&vP_1\end{pmatrix}$ is a paraunitary matrix. 




\subsection{General construction}\label{grty}
Let $\{E_0,E_1, \ldots, E_k\}$ be a complete symmetric
orthogonal set of idempotents in $F_{n\ti n}$. Arrange these into a 
$k\ti k$ block matrix of with each 
 row of blocks containing one of the blocks 
$\{E_0, E_1, \ldots , E_k\}$ exactly once. Now attach monomials to each
$E_i$; the same
monomial need not be used with each $E_i$ that appears. Let $W$ be the
resulting matrix.
\begin{theorem} $W$ is a paraunitary matrix.
\end{theorem}
\begin{proof} Take the block inner product of two different rows of
 blocks. The $E_i$ are orthogonal to one another so the result is
 $0$. Take the block inner product of any row of blocks with
 itself. This gives 
$E_1^2+E_2^2+\ldots + E_k^2 = E_1+E_2+\ldots + E_k = 1 (=I_n)$. Hence
 $WW^* = 1 ( =I_{nk})$.
\end{proof}
  
The $W$ is a paraunitary matrix in the union of the
variables of the monomials.   

The condition that each row and column block contains each $E_i$
once and once only can be obtained by using the group ring matrix of any
group of order $k$; see for example \cite{hur3}. 
So for example the $E_i$ could be arranged as a
circulant block of matrices. Different arrangements will in general give
inequivalent paraunitary matrices.

(Modifications in the construction of $W$   
by attaching elements of modulus $1$ as coefficients will give
paraunitary matrices but further
conditions are necessary on these elements.These modifications are not
 included here.)

These are nice constructions for paraunitary matrices but can be shown 
to be separable.  However  they have uses in themselves and will 
prove useful later as parts of 
constructions of non-separable paraunitary matrices. 
They may also be used to construct special and regular 
types of Hadamard real and complex matrices. These are illustrated in the
following sections \ref{two} and \ref{three4}.

\subsection{Monomials and Hadamard regular matrices}\label{two}
Although the matrices of idempotents as constructed in sections
\ref{gyt} and \ref{grty} 
produce separable paraunitary matrices these can
be useful structures in themselves; they may also be used in certain cases 
to produce `regular' Hadamard matrices. (Walsh-Hadamard matrices are
 regular Hadamard matrices which have been used extensively
 in communications' theory.) 
If in a paraunitary matrix the entries are $\pm 1$ times monomials 
in the variables then substituting $\pm 1$ for the variable gives a
Hadamard matrix. If in a paraunitary matrix the entries are $\om$ times
monomials where $\om$ is a complex number of modulus $1$ then
substituting each variable by a complex number of modulus $1$ gives a
complex Hadamard matrix. 

For example use $P_0= \cir(1,1,1), P_1=\cir(1,\om,\om^2), P_2=
\cir(1,\om^2, \om)$ where $\om$ is a primitive third root of unity gives
the following matrix: $\begin{pmatrix}xP_0& yP_1 & zP_2 \\ zP_2 & xP_0 & yP_1 \\ yP_1
	   & zP_2 & xP_0 \end{pmatrix}$.

Substituting a complex number of modulus $1$ for each of $x,y,z$ gives a
complex Hadamard matrix.

Butson-type Hadamard matrices $H(q,n)$ are complex Hadamard $n\ti n$ 
matrices with entries which are $q^{th}$ roots of unity.

With the above example substituting a third
root of unity for the variables gives a Hadamard $H(3,9)$ matrix, that
is a a matrix $H$ with entries which are third roots of unity so that
$HH^*= 9$. These matrices could then be used to 
produce  Hadamard $H(3,36)$ matrices.

This can  be extended to $q\ti q$ matrices using the complete
orthogonal set of idempotents for the
cyclic group of order $q$; this will involve $q^{th}$ roots of
unity and is related to the representation theory of the finite cyclic
group.  From this Hadamard $H(q,q^2)$
matrices can be produced and from these Hadamard $H(q,(2q)^2)$ matrices
can be produced and so on.


\subsection{Mixing}\label{three4}
It is noted that interchanging rows and/or columns in a paraunitary
matrix results in a paraunitary matrix and thus interchanging blocks of
rows and/or columns results in a paraunitary matrix. When using 
complete orthogonal sets of idempotents the blocks of
row and/or columns of idempotents can be interchanged in the
construction stage. The resulting
paraunitary matrices can take a regular form. In certain cases the
variables can be specialised to form regular Hadamard matrices. Here 
examples are given but details of the constructions are omitted.
   
Let $\{E_0,E_1\}$ be a complete symmetric orthogonal set of idempotents in
  $\mathbb{C}_{n\ti n}$. For
  example in $\mathbb{C}_{2\ti 2}$ these could be $E_0=\frac{1}{2}
  \begin{ssmatrix} 1&1\\1&1\end{ssmatrix},
  E_1=\frac{1}{2}\begin{ssmatrix} 1&-1 \\ -1&1\end{ssmatrix}$.

Define 

$W= \frac{1}{4}\begin{ssmatrix}xE_0 & yE_1 & zE_0& tE_1 &xE_0
&  yE_1 & -zE_0& -tE_1 \\ -uE_1& vE_0 & wE_1& pE_0 & uE_1&
  vE_0&wE_1&-pE_0\\ zE_0&-tE_1& xE_0&yE_1 & -zE_0&tE_1&xE_0&yE_1 \\
  -wE_1&pE_0&-uE_1&vE_0 & -wE_1&-pE_0&uE_1&vE_0 \\
  xE_0&yE_1&-zE_0&-tE_1& xE_0&yE_1&zE_0&tE_1 \\
  uE_1&vE_0&wE_1&-pE_0&-uE_1&vE_0&wE_1 &pE_0\\
  -zE_0&tE_e&xE_0&yE_1&zE_0&-tE_1&xE_0 & yE_1\\
  -wE_1&-pE_0&uE_1&vE_0&-wE_1&pE_0&-uE_1&vE_0 \end{ssmatrix}$. 

Then $W$ is a paraunitary matrix in the variables $\{x,y,z,t,u,v,w,p\}$. 

Interchanging blocks of rows and/or columns in $W$ will also give a
paraunitary matrix which is equivalent to $W$. 
By varying the signs
other constructions of paraunitary matrices are obtained and these 
 are not generally equivalent to one another.

They may  seem to be non-separable but  by interchanging 
blocks of rows and blocks of columns it can be shown that they are 
 separable. However they 
are interesting in themselves with interesting properties and can also
be used to construct Hadamard matrices of regular types. 

 Examples 
 may also be interpreted as coming  from
  the structure of group rings of various groups. Then giving the values
  $\pm 1$ to the variables can result in regular Hadamard matrices or
  giving the values $e^{i\theta}$ to the variables can result in Hadamard
  complex matrices. Hadamard matrices with entries which are roots of
  unity may also be obtained from these constructions.

The following is an example of this type.

$$w(x,y,z,t)=\frac{1}{4}\begin{ssmatrix} x & x 
& y& -y & z& z& t& -t & x& x & y& -y
	     & -z & -z & -t & t \\ x& x& -y & y & z& z & -t & t & x& x &
	     -y & y & -z & -z & t & -t \\ - t& t & x& x & y & -y & z & z
	     & t & -t & x & x & y & -y & -z & -z \\ t & -t & x & x & -y
	     & y & z & z & -t & t & x & x & -y & y & -z & -z \\ z&z & -t
	     & t & x & x & y & -y & -z & -z & t & -t & x & x & y & -y
	     \\ z & z & t & -t & x & x & -y & y & -z & -z & -t & t & x
	     & x & -y & y \\
	     -y & y & z & z & -t & t & x & x & -y & y & -z & -z & t & -t
	     & x & x \\ y& -y & z& z & t & -t & x & x & y & -y & -z & -
	     z & -t & t & x & x \\  x & x & y & -y & -z & -z & -t & t & x
	     & x & y & -y & z & z & t & -t \\ x& x & -y & y & -z & -z &
	     t & -t & x& x & -y & y & z & z & -t & t \\ t & -t & x & x &
	     y& -y & -z & -z & -t & t & x & x & y & -y & z & z \\ -t & t
	     & x & x & -y & y & -z & -z & t & -t & x& x & -y & y & z & z
	     \\ -z & -z & t & -t & x & x & y & -y & z & z & -t & t & x &
	     x & y & -y \\-z & -z & -t & t & x & x & -y & y & z & z & t
	     & -t & x & x & -y & y \\ -y & y & -z & -z & t & -t & x & x&
	     -y & y & z & z & -t & t & x & x \\ y & -y & -z & -z & -t &
	     t & x & x & y & -y & z & z & t & -t & x & x
	    \end{ssmatrix}$$

This matrix has the form $\begin{pmatrix} P & Q
						  \\ Q & P 
						 \end{pmatrix}$
where $PQ^*=0 = QP^*$ and has the structure of the group ring of {$C_8\cross C_2$}.

It is clear that $x,y,z,t$ may each be replaced by a 
monomial times a complex number of modulus $1$ in $W$ and a paraunitary
matrix is obtained. 

Values of $\pm 1$ may be given to the variables in $W$ and with the
fraction omitted this gives a regular real Hadamard matrix. Values of modulus
$1$ may be given to to the variables in $W$ and with the fraction
omitted a regular complex Hadamard matrix is obtained. Other group ring
structures can also arise in this manner. 


Walsh-Hadamard matrices have the  structure of the group
ring of $C_2^n$. These examples may also be extended in a similar
way. So for example 
constructions with the structure of the group ring 
 $C_8\cross C_2\cross G$ may be made if paraunitary matrices with the
 structure of $G$ can be formed such as when $G= C_4^n$ or when $G=(C_2\cross
 C_2)^n$ and others.

\section{Non-separable constructions}\label{nonsep}

Several methods have now been developed for
constructing multidimensional paraunitary methods from complete
orthogonal sets of matrices. The matrices produced  are useful in many ways but
 are found to be separable although not
appearing so initially. The methods use just one complete orthogonal
set of idempotents in the various constructions. 

Thus we are led to consider {\em different} complete orthogonal sets
of variables and to `tangle' them up in order to construct non-separable
paraunitary matrices.



\subsection{A general construction}
\begin{proposition}\label{fty1} Let $A,B$ be paraunitary matrices of the same size
 but not necessarily
 with the same variables over a field in which $2$ has a square
 root. Then $W=\frac{1}{\sqrt{2}}\begin{pmatrix}  A&B \\
				  A&-B\end{pmatrix}$ and 
$Q=\frac{1}{\sqrt{2}}\begin{pmatrix} A&A \\
				  B&-B\end{pmatrix}$  
is a  paraunitary matrix in the union of the variables in $A,B$.
\end{proposition}

\begin{proof} Suppose $A,B$ are $n\ti n$ matrices. Then \begin{eqnarray*}
WW^*&=&\frac{1}{2}\begin{pmatrix} A&B \\
				  A&-B\end{pmatrix}\begin{pmatrix} A^*&A^* \\
				  B^*&-B^*\end{pmatrix} \\ &=&
				  \frac{1}{2}\begin{pmatrix}AA^*+BB^*&AA^*-BB^* \\
				   AA^*-BB^* & AA^*+BB^*\end{pmatrix} \\
		   &=& \frac{1}{2}\begin{pmatrix} I_n + I_n & I_n-I_n \\
				   I_n-I_n & 
			I_n+I_n
\end{pmatrix} = I_{2n} \end{eqnarray*}

The case for $Q$ can be considered similarly or it follows from Lemma
 \ref{trans} since $Q$ is the transpose of $W$.
\end{proof}

Paraunitary matrices constructed by methods of previous sections 
may be used as input to Proposition 
\ref{fty1} to  construct paraunitary matrices.
Matrices constructed using the Proposition can then also be used
as input. 

The methods are fairly general and it is easy to produce 
examples for input using various complete orthogonal sets of idempotents.
The result holds in general over any field which contains the square
root of $2$.

If $A=B$ then $W$ in Proposition \ref{fty1} is the tensor product
$A\otimes J$ where $J=\begin{ssmatrix} 1 & 1 \\ 1& -1 \end{ssmatrix}$.
If $A$ and $B$ are formed using the same complete symmetric orthogonal set of
idempotent as in \ref{multi} or \ref{gyt} then $W$ can be shown to be separable.

It would appear initially that Proposition \ref{fty1} 
could/should be generalised to 
$W=\frac{1}{\sqrt{2}}\begin{pmatrix} pA& qB \\ pA &-qB
				    \end{pmatrix}$ where $p,q$ are
      monomials   or in the case of $\cc$ 
where $p,q$ are monomials times
      a complex number of modulus $1$. However then $W$ is separable as
      a product  
$W=\frac{1}{\sqrt{2}}\begin{pmatrix} A&B \\
				  A&-B\end{pmatrix}\begin{pmatrix}pI & 0 \\ 0 &
						   qI\end{pmatrix}$.

If $W=\begin{pmatrix} X & Y \\ Z & T\end{pmatrix}$ where
$X,Y,Z,T$ are matrices of the same size then $X,Y,Z,T$ are referred to
as the {\em blocks} of $W$ and $\begin{pmatrix} X & Y \end{pmatrix}$ and
$\begin{pmatrix}Z & T\end{pmatrix}$ are the {\em row blocks} of
$W$. Similarly {\em column blocks} of $W$ are defined.
						   

Suppose $A,B$ are matrices of the same size. Then a {\em tangle} of
$\{A,B\}$ is  one of
\begin{enumerate}
\item $W=\frac{1}{\sqrt{2}}\begin{pmatrix} A& B \\ A &-B
				    \end{pmatrix}$.
\item A matrix obtained from 1.\  by interchanging rows of
      blocks and/or columns of blocks.
\item The transpose of any matrix obtained in  1.\ or 2.
\end{enumerate}

A tangle of $\{A,B\}$ is not the same as,  and is  
 not necessarily equivalent to, a tangle of $\{B,A\}$.
Note that interchanging any rows and/or columns of a paraunitary matrix
results in an (equivalent) paraunitary matrix. Thus in particular
interchanging rows and/or columns of blocks also results in equivalent
paraunitary matrices; thus item 2.\ gives  equivalent paraunitary
matrices to item 1. The negative of a paraunitary matrix is a paraunitary
matrix as is the $^*$ of a paraunitary matrix.

For example  \\ $\frac{1}{\sqrt{2}}\begin{pmatrix} A& B \\ A &-B
				    \end{pmatrix},
\frac{1}{\sqrt{2}}\begin{pmatrix} A& -B \\ A &B
				    \end{pmatrix},
\frac{1}{\sqrt{2}}\begin{pmatrix} A& A \\ B & -B
				    \end{pmatrix}$
are tangles of $\{A,B\}$
\\ and \\
$\frac{1}{\sqrt{2}}\begin{pmatrix} B& A \\ B &-A
				    \end{pmatrix},
\frac{1}{\sqrt{2}}\begin{pmatrix} B& -A \\ B &A
				    \end{pmatrix},
\frac{1}{\sqrt{2}}\begin{pmatrix} B& B \\ A & -A
				    \end{pmatrix}$
are tangles of $\{B,A\}$.


 Proposition \ref{fty1} may be generalised as follows.
\begin{proposition}\label{fty} Let $A,B$ be paraunitary matrices of the
 same size 
 but not necessarily
 with the same variables. Then a tangle of $\{A,B\}$ or $\{B,A\}$ is a paraunitary
 matrix.
\end{proposition}

Use the expression `$A$ is tangled with $B$' to mean that a tangle of
$\{A,B\}$ or $\{B,A\}$ is formed.

\subsection{Examples}
\begin{enumerate} \item \begin{enumerate} \item 
Construct $A = (x)$ and $B=(y)$. \item Construct $W=\frac{1}{\sqrt{2}}\begin{ssmatrix} A& B \\
A & -B\end{ssmatrix} = \frac{1}{\sqrt{2}}\begin{ssmatrix} x & y \\ x & -y\end{ssmatrix}$.
			 Then $W$ is a paraunitary matrix. \item
								Similarly 
			 construct $Q=\begin{ssmatrix} z & t \\ z &
				    -t\end{ssmatrix}$.
\item Tangle $W$ and $Q$ to produce for example the paraunitary matrix  
$T=\frac{1}{2}\begin{ssmatrix} x& y& z& t \\ x&-y& z& -t \\ x&y&-z& -t \\ x&-y&-z& y
 \end{ssmatrix}$.
\item The process can be continued: Matrices produced from (d), with
      different variables, can be input to form further paraunitary
      matrices.

\end{enumerate} 
\item \begin{enumerate}
\item Construct as is \ref{finite} the following complete symmetric sets
      of idempotents in $3\ti 3 $ matrices over $\mathbb{F}_7$:

 $\{P_0=\begin{ssmatrix} 2&1&2\\ 1&4&1 \\ 2&
1&2 
\end{ssmatrix},
P_1=\begin{ssmatrix} 4&1&6 \\ 1&2& 5 \\ 6&5& 2 \end{ssmatrix},
P_2\begin{ssmatrix} 2 & 5&6 \\ 5&2& 1 \\ 6 & 1 &4\end{ssmatrix}\}$,
$\{ Q_0=\begin{ssmatrix} 6&5&6 \\ 5&3&5 \\ 6& 5&6 \end{ssmatrix},
Q_1=\begin{ssmatrix} 5&2&5 \\ 2&5& 2 \\ 5&2& 5 \end{ssmatrix},
Q_2= \begin{ssmatrix} 4 & 0&3 \\ 0&0 & 0 \\ 3 & 0 &4\end{ssmatrix}\}$.
\item Form $A=xP_0+yP_1+zP_2, B=tQ_0+rQ_1+sQ_2$.
\item Tangle $A,B$ to form for example the following paraunitary matrix
      over $\mathbb{F}_7$: $\begin{ssmatrix} A & A \\ -B &
			    B\end{ssmatrix}$.
\end{enumerate}
\item 
\begin{enumerate}
\item Construct, in $\mathbb{C}_{2\ti 2}$, the following complete symmetric
      (different) sets of orthogonal idempotents 
$\{E_0,E_1\}$ and $\{Q_0,Q_1\}$ where:
 
$E_0=\frac{1}{2}\begin{ssmatrix} 1 & 1 \\ 1&1
  \end{ssmatrix}, \, E_1=\frac{1}{2}\begin{ssmatrix} 1&-1 \\ -1&1
  \end{ssmatrix}.$
\hspace{.1in} $Q_0=\frac{1}{5}\begin{ssmatrix}4 & 2 \\ 2 & 1 \end{ssmatrix}, \,
Q_1 = \frac{1}{5}\begin{ssmatrix}1&-2 \\ -2 & 4 \end{ssmatrix}.
$ 
\item Construct $A=xE_0+yE_1, B=zQ_0+tQ_1$.
\item Construct $W=\frac{1}{\sqrt{2}}\begin{pmatrix} A&B \\
				  A&-B\end{pmatrix}$. Then 
$W$ is a paraunitary matrix of size $4\ti 4$ with variables
      $\{x,y,z,t\}$.

\end{enumerate} 
\item \begin{enumerate}
\item Construct different $\{E_0,E_1\}$ and $\{Q_0,Q_1\}$ 
complete symmetric  orthogonal sets of idempotents in $\cc_{n\ti n}$.  
\item Construct $W= \frac{1}{\sqrt{2}}\begin{pmatrix} xE_0 & yE_1 & uQ_0&vQ_1
  \\ rE_1&pE_0 & zQ_1&tQ_0 \\ xE_0&yE_1 & -uQ_0 & -v Q_1 \\
  rE_1&pE_0&-zQ_1&-tQ_0 \end{pmatrix}$. Then $W$ is a paraunitary
  matrix.

It is essential that $\{E_0,E_1\}$ and $\{Q_0,Q_1\}$ are {\em different}
complete orthogonal sets of idempotents in order for $W$ to be
non-separable although the construction does not depend on this.

\item Clearly also the roles of $\{E_0,E_1\}$ and $\{Q_0,Q_1\}$ can be
 interchanged in $W$ and a paraunitary matrix is still obtained. 
Changing the $\pm$ signs 
in such a way that the block inner product of any two rows of blocks is
$0$ will give a different inequivalent paraunitary matrix. Thus for example 
$W$ could be replaced by the following: 
\item Construct $W= \frac{1}{\sqrt{2}}\begin{pmatrix} xE_0 & yE_1 & uQ_0&vQ_1
  \\ rE_1&pE_0 & -zQ_1&-tQ_0 \\ -xE_0&-yE_1 & uQ_0 & v Q_1 \\
  rE_1&pE_0&zQ_1&tQ_0 \end{pmatrix}$. Then this $W$ is also a
      paraunitary matrix.

\end{enumerate}

\item \begin{enumerate} \item 
 See section \ref{idems} for methods for constructing  
complete orthogonal symmetric sets of idempotents. 
Examples  of such sets  in $\cc_{2\ti 2}$ are $\{E_0,E_1\}, \{Q_0,Q_1\}$
where these are given as in Example 2 above.  


\item Another complete symmetric orthogonal set of idempotents in $\cc_{2\ti 2}$
is the following: \\
$\{P_0 =
\frac{1}{2}\begin{ssmatrix} 1 &-i \\ i & 1 \end{ssmatrix}, P_1 =
\frac{1}{2}\begin{ssmatrix}1 &i \\ -i & 1 \end{ssmatrix}\}$.
\item In example 2. $\{E_0,E_1\}$ is `tangled' with $\{Q_0,Q_1\}$.
$\{P_0,P_1\}$ may similarly be combined (`tangled') with either $\{E_0,E_1\}$ or
$\{Q_0,Q_1 \}$ to construct paraunitary matrices. 

\item Using  $\{P_0,P_1\}$
with $\{E_0,E_1\}$ produces paraunitary matrices of the form
$\frac{1}{2} P$ where the entries of $P$ are $\pm 1, \pm i$ with
$i=\sqrt{-1}$. By specialising the variables, complex Hadamard
matrices may be  obtained. 

\item For example the following is a
paraunitary matrix:
$\frac{1}{2\sqrt{2}}\begin{ssmatrix} x & x & y & -y& u & -iu & v & iv \\ x&x & -y & y
&iu&u&-iv& v \\ r&-r & p& p & z & iz& t & -it \\ -r & r& p&p& -iz& z& it& t
\\ x & x & y&-y & -u & iu & -v & -iv \\ x&x&
-y&y&-iu&-u&iv&-v\\ r&-r&p&p&-z&-iz&-v&iv \\ -r&r&p&p&iz&-z&-iz& z\end{ssmatrix}$. 

\item By giving values of modulus $1$ to the variables, complex Hadamard
 matrices are obtained. For example letting all the variables have the
 value $+1$ gives the following complex Hadamard matrix: \\
$\begin{ssmatrix} 1 & 1 & 1 & -1& 1 & -i & 1 & i \\ 1&1 & -1 & 1
&i&1&-i& 1 \\ 1&-1 & 1& 1 & 1 & i& 1 & -i \\ -1 & 1& 1&1& -i& 1& i& 1
\\ 1 & 1 & 1&-1 & -1 & i & -1 & -i \\ 1&1&
-1&1&-i&-1&i&-1\\1&-1&1&1&-1&-i&-1&i \\ -1&1&1&1&i&-1&-i& 1\end{ssmatrix}$. 
\end{enumerate}

\item \begin{enumerate} \item 
Construct  $P_0 =\frac{1}{9}\begin{ssmatrix} 4 & 2 & 4 \\ 2&1&2 \\ 4 &2
			      &4 \end{ssmatrix},  P_1=
			      \frac{1}{9}
\begin{ssmatrix} 1 & 2 & -2 \\ 2&4&-4 \\
			       -2 &-4 & 4\end{ssmatrix}, 
P_2 =  \frac{1}{9}\begin{ssmatrix}4 & -4&-2 \\ -4 &4 & 2
			      \\ -2 & 2& 1\end{ssmatrix}$

\item Construct  the complete symmetric orthogonal set of idempotents obtained
from the group ring $\cc C_3$ of the cyclic group of order $3$:
 $Q_0= \cir(1,1,1), Q_1=\cir(1,\om,\om^2), Q_2 =
\cir(1,\om^2,\om)$ where $\om$ is a primitive cube root of $1$; 

\item Construct $A= \begin{pmatrix}xP_0 & yP_1 & z P_2 \\ pP_2 & qP_0 & r
P_1 \\ sP_1&tP_2&vP_0\end{pmatrix}$ and 
$B=\begin{pmatrix}aQ_0& bQ_1 & cQ_2 \\ dQ_2 & eQ_0 & fQ_1 \\ gQ_1
	   & hQ_2 & kQ_0 \end{pmatrix}$.

\item Construct $W=\frac{1}{\sqrt{2}}\begin{pmatrix} A&B \\
				  A&-B\end{pmatrix}$.

Then $W$ is a paraunitary matrix. It has size $18\ti 18$ and $18$
variables;  variables can be equated.
 \end{enumerate}
\end{enumerate}

$W$ in the above could for example be replaced   
by $W=\frac{1}{\sqrt{2}}\begin{pmatrix} ipA&qB \\
				  ipA&-qB\end{pmatrix}$ where $p,q$ are
      variables and $i=\sqrt{-1}$ but as pointed out 
this is separable and may be constructed as a product.

A complex Hadamard matrix is a matrix $H$ of size $n\ti n$ with entries
of modulus $1$ and satisfying $HH^*=nI_n$. 
Complex Hadamard matrices arise in the study of operator algebras and
 the theory of quantum computation. 

By giving values which are $k^{th}$
 of unity  to the variables, with  $k$  divisible by $4$, 
in the above example 4, special types of complex Hadamard matrices
 which are called {\em Butson-type} are obtained. A Butson type Hadamard
 $H(q,n)$ matrix is a complex Hadamard matrix of size $n\ti n$ all of
 whose entries are $q^{th}$ roots of unity. 
 \medskip \medskip

Here is another example which uses group rings:.

\begin{enumerate}\item 
Construct the  complete symmetric set of
orthogonal idempotents $\{P_i | i= 0,1,\ldots, 5\}$ 
from the group ring $\cc C_6$ of the cyclic
group $C_6$ of order $6$. This gives
$P_i=\cir(1,\om^i,\om^{2i},\om^{3i},\om^{4i},\om^{5i})$ where $\om$ is
a primitive 6th root of unity. 

\item Define $Q_0=P_0, Q_1=P_1+P_5, Q_2=P_2+P_4$. Note that $Q_0,Q_1,Q_2$ are
real.

\item Let $$E_1=\frac{1}{6} \left(\begin{smallmatrix} 1 & 1&1&1&1&1\\
  1&1&1&1&1&1 \\ 1&1&1&1&1&1 \\ 1&1&1&1&1&1 \\ 1&1&1&1&1&1
  \\ 1& 1&1&1&1&1 \end{smallmatrix}\right),
E_2=\frac{1}{6} \left(\begin{smallmatrix}1 & -1&-1&-1&1&1\\
  -1&1&1&1&-1&-1 \\ -1&1&1&1&-1&-1 \\ -1&1&1&1&-1&-1 \\ 1&-1&-1&-1&1&1
  \\ 1& -1&-1&-1&1&1 \end{smallmatrix}\right),
E_3=\frac{1}{3} \left(\begin{smallmatrix} 2 & 0&0&0&-1&-1\\
  0&2&-1&-1&0&0 \\ 0&-1&2&-1&0&0 \\ 0&-1&-1&2&0&0 \\ -1&0&0&0&2&-1
  \\ -1& 0&0&0&-1&2 \end{smallmatrix}\right).$$

be the complete symmetric set of orthogonal idempotents obtained from
the group ring of 
$S_3 (=D_6)$ as in section \ref{symm6}.

\item Define 
$A= \begin{pmatrix}xE_0 & yE_1 & z E_2 \\ pE_2 & qE_0 & r
E_1 \\ sE_1&tE_2&vE_0\end{pmatrix}$ and 
$B=\begin{pmatrix}aQ_0& bQ_1 & cQ_2 \\ dQ_2 & eQ_0 & fQ_1 \\ gQ_1
	   & hQ_2 & kQ_0 \end{pmatrix}$.

\item Construct $W=\frac{1}{\sqrt{2}}\begin{pmatrix} A&B \\
				  A&-B\end{pmatrix}$.

Then $W$ is a paraunitary matrix which is real.
It is a $36\ti 36$ matrix with $18$ variables.
\end{enumerate}
\subsection{An Algorithm}\label{sisty}
\begin{enumerate}
\item Construct different sets $\{ P_0,P_2, \ldots, P_k\}$ and $\{ Q_0,Q_1,
      \ldots, Q_k\}$ of complete orthogonal symmetric of idempotents by
      the methods of section \ref{joint} or section \ref{grring} (or otherwise) 
in $F_{n\ti n}$. ($F$ is usually
      $\cc$ but other fields can also be used.)
\item Construct a paraunitary matrix from  $\{ P_0,P_2, \ldots, P_k\}$
      by either the methods of section \ref{multi} or the methods of 
section \ref{gyt}. Call
      this matrix $A$.
\item Construct a paraunitary matrix from $\{ Q_0,Q_1,
      \ldots, Q_k\}$
      by either the methods of section \ref{multi} or the methods of 
section \ref{gyt}. Call
      this matrix $B$. The variables in $A,B$ can be different. 
\item Construct a tangle of $\{A,B\}$ or $\{B,A\}$.

\end{enumerate}


 \vspace{.1in}
An example: 

\begin{enumerate} 
\item Construct $\{P_0,P_1,P_2,P_3\}$ and $\{Q_0,Q_1,Q_2,Q_3\}$ in
 $\cc_{4\ti 4}$ where the $P_i$ and $Q_j$ are obtained by construction
 methods of \ref{joint}, \ref{grring}.
\item Form $\begin{pmatrix} P_0&P_1&P_2&P_3 \\ P_3 & P_0&P_1&P_2 \\
	      P_2&P_3&P_0&P_1 \\ P_1& P_2&P_3&P_0 \end{pmatrix}$. (Here
      the structure of $C_4$ is used and a circulant structure is obtained.)
\item Form $\begin{pmatrix} x_{01}P_0&x_{11}P_1&x_{21}P_2&x_{31}P_3 \\ 
x_{02}P_3 & x_{12}P_0&x_{22}P_1&x_{32}P_2 \\
	      x_{03}P_2&x_{13}P_3&x_{23}P_0&x_{33}P_1 \\ x_{04}P_1&
	     x_{14}P_2&x_{24}P_3&x_{34}P_0 \end{pmatrix}$.
\item Let the matrix in 3.\ be denoted by $A$.
\item Form $\begin{pmatrix} Q_0&Q_1&Q_2&Q_3 \\ Q_1 & Q_0&Q_3&Q_2 \\
	      Q_2&Q_3&Q_0&Q_1 \\ Q_3& Q_2&Q_1&Q_0 \end{pmatrix}$. (Here
      the structure of $C_2\cross C_2$ is used.)
\item Form  $\begin{pmatrix} y_{01}Q_0&y_{11}Q_1&x_{21}Q_2&x_{31}Q_3 \\ 
y_{02}Q_1 & y_{12}Q_0&y_{22}Q_3&y_{32}Q_2 \\
	      y_{03}Q_2&y_{13}Q_3&y_{23}Q_0&y_{33}Q_1 \\ y_{04}Q_3&
	     y_{14}Q_2&y_{24}Q_1&y_{34}Q_0 \end{pmatrix}$.
\item Let the matrix in 6.\ be denoted by $B$.
\item Form $W= \begin{pmatrix} A & B \\ A & -B\end{pmatrix}$. 
\end{enumerate}
\subsection{Further algorithm}\label{further}
\begin{enumerate}
\item Input paraunitary matrices $A,B$ of the same size but not
      necessarily with the same variables. These may be formed from
      method of section \ref{sisty} or from this algorithm.
\item Form a tangled product of $\{A,B\}$ or $\{B,A\}$.

\end{enumerate}

\subsection{Further constructions}

The non-separable paraunitary matrices and separable paraunitary
matrices can be combined when appropriate 
as products or as tensor products  to construct further
paraunitary matrices. These may then be input to algorithm of section
\ref{further}. 

  

\section{Pseudo-paraunitary}\label{pseudo}
Let $P$ be a paraunitary $n\ti n$ matrix with variables ${\bf z}$ over a field $F$.
Then the rows $\{v_1,v_2, \ldots, v_n\}$ of $P$ satisfy $v_iv_i^* = 1$ and
$v_iv_j^*=0$ for $i\neq j$. Note that $v^*$ means 
transpose conjugate over $\cc$, and transpose over other fields, 
with the understanding that 
$z^*=z^{-1}, \{z^{-1}\}^*=z$ for a variable $z$.

Let $P_i=v_i^*v_i$ which are $n\ti n$ matrices of rank $1$ and
involve the variables $\{\bf z, z^{-1}\}$. Then $\{P_1, P_2, \ldots, P_n\}$ is
a complete orthogonal symmetric set of idempotents in the polynomial
ring $F_{n\ti n}[{\bf z, z^{-1}}]$. Hence by the methods of Sections
\ref{idems}, \ref{gyt}, \ref{multi}
  and \ref{nonsep} paraunitary-type matrices may be formed; for the
  method of Section \ref{nonsep} two such sets must be constructed. 
For example
 $W=w_1P_1+w_2P_2+\ldots + w_nP_n$ in
  variables ${\bf w}= (w_1,w_2,\ldots, w_n)$ satisfies $WW^*=1$. Now
  $W$ is a matrix  in the variables $(\bf z,z^{-1},w)$ but cannot be
  termed {\em paraunitary}. Call such a matrix a {\em pseudo-paraunitary}
  matrix. Having constructed $W$ its rows may then be used to
  construct further pseudo-paraunitary matrices and so on.

Thus say $W({\bf z,  z^{-1}})\in F_{n\ti n}[{\bf z, z^{-1}}]$ 
is a {\em pseudo-paraunitary} matrix
if $WW^*=1$. Pseudo-paraunitary matrices may be constructed from
paraunitary matrices and from pseudo-paraunitary matrices.

Here's an  example. 
\begin{enumerate}
\item Form $E_0=\frac{1}{2}\begin{ssmatrix} 1 & 1 \\ 1 &
  1\end{ssmatrix},
E_1=\frac{1}{2}\begin{ssmatrix} 1&-1 \\ -1 & 1 \end{ssmatrix}$.
\item Form $P= xE_0+yE_1 = \frac{1}{2}\begin{ssmatrix} x+y & x-y \\
  x-y & x+y \end{ssmatrix}$ 
\item Let $v_1=\frac{1}{2}(x+y, x-y), v_2=\frac{1}{2}(x-y, x+y)$.
\item Form $P_1=v_1^*v_1 = \frac{1}{4}\begin{ssmatrix} 2 +
  x^{-1}y+y^{-1}x & y^{-1} x-x^{-1}y \\ x^{-1}y-y^{-1}x &
  2-x^{-1}y-y^{-1}x \end{ssmatrix},
P_2=v_2^*v_2 = \frac{1}{4}\begin{ssmatrix} 2 -
  x^{-1}y-y^{-1}x & y^{-1} x-x^{-1}y \\ x^{-1}y-y^{-1}x &
  2+x^{-1}y+y^{-1}x \end{ssmatrix}$.
\item $P_iP_i=P_i, P_1P_2=0$ and $P_1+P_2=1$. Form $W=zP_1+tP_2$. 
\item $W = W(x,y,x^{-1}, y^{-1}, z,t)$, and $WW^*= 1$.
\item The rows of $W$ can be used  to form further 
  pseudo-paraunitary matrices.
\end{enumerate}
 
Consider $P_1,P_2$ as in this example above. Define $Q_1 = xyP_1,
Q_2=xyP_2$. Define $Q= zQ_1+tQ_2$. Then $Q\in F_{2\ti 2}[x,y,x,t]$,
and is  a polynomial with $QQ^*= x^2y^2I_2$.

Say $W\in F_{n\ti n}[\bf z]$ is a {\em pseudo-paraunitary} matrix over
$F_{n\ti n}[\bf z]$ if
$WW^*=pI_n$ where $p$ is a monomial in $\bf z$. A
pseudo-paraunitary matrix over $F_{n\ti n}[{\bf z,z^{-1}}]$ may be used
to construct a pseudo-paraunitary matrix over $F_{n\ti n}[\bf z]$ and
vice versa.

Pseudo-paraunitary matrices in general may be constructed
from paraunitary matrices and from pseudo-paraunitary matrices.

\section{Determinants and rank}\label{grmat}
Here we consider properties of complete sets of idempotent matrices
and ranks of the idempotents. 
 

\begin{lemma}\label{trrank} Suppose $\{E_1,E_2, \ldots, E_s\}$ is a
set of orthogonal idempotent matrices. Then $\rank
(E_1+E_2 +\ldots + E_s) = \tr (E_1+E_2+ \ldots + E_s) = \tr E_1+ \tr
 E_2+ \ldots + \tr E_s = \rank E_1+ \rank E_2 + \ldots +\rank
E_s$.
\end{lemma}
\begin{proof}
It is known that $\rank A = \tr A$ for an idempotent matrix, see
for example \cite{idemrank}, and so
$\rank E_i = \tr E_i$ for each $i$. If $\{E,F,G\}$ is a set an orthogonal
 idempotent matrices so is  $\{E+F,G\}$. From this it follows that $\rank
(E_1+E_2 +\ldots + E_s) = \tr (E_1+E_2+ \ldots E_s)= \tr E_1+\tr E_2 +
 \ldots + \tr E_s = \rank E_1+ \rank E_2 + \ldots \rank
E_s$.
\end{proof}
\begin{corollary}\label{trrank1}
$\rank(E_{i_1}+ E_{i_2}+ \ldots + E_{i_k})= 
\rank E_{i_1} +\rank E_{i_2}+ \ldots + \rank E_{i_k}$ for $i_j \in \{
1,2,\ldots, s\}$, and $i_j\neq i_l$ for $j\neq l$.
\end{corollary}

Let $\{e_1, e_2, \ldots, e_k\}$ be a complete orthogonal set of idempotents
in a vector space over $F$. 


\begin{theorem}\label{gr1} Let $w= \al_1 e_1 + \al_2 e_2 + \ldots +
\al_ke_k$ with $\al_i \in F$. Then $w$
 is invertible if and only if each $\al_i \neq 0$ 
and in this case $w^{-1}
 = \frac{1}{\al_1}e_1 + \frac{1}{\al_2}e_2+ \ldots + \frac{1}{\al_k}e_k$.
\end{theorem} 

\begin{proof} Suppose each $\al_i \neq 0$.
Then $w(\frac{1}{\al_0}e_0+\frac{1}{\al_1}e_1+ \ldots + \frac{1}{\al_k}e_k)
 = e_0^2 + e_1^2 + \ldots + e_k^2 = e_0+e_1+\ldots + e_k = 1$.

Suppose $w$ is invertible and that some $\al_i=0$. 
Then $we_i =0$ and so $w$ is a (non-zero) zero-divisor and is not invertible.
\end{proof}


We now specialise the $e_i$ to be $n\ti n$ matrices and in this case
use capital letters and let $e_i = E_i$.
 
Let $A= a_1 E_1 + a_2 E_2 + \ldots + a_kE_k$. Then $A$
 is invertible if and only if each $a_i \neq 0$ and in this case $A^{-1}
 = \frac{1}{a_1}E_1 + \frac{1}{a_2}E_2+ \ldots + \frac{1}{a_k}E_k$. 

\begin{theorem}{\label{det}} Suppose $E_1, E_2, \ldots, E_k$ is a
 complete symmetric orthogonal set of idempotents in $F_{n\ti n}$. 
Let $A= a_1 E_1 + a_2 E_2 + \ldots +
 a_kE_k$ with $a_i\in F$.  Then the determinant of $A$ is 
$|A| = a_1^{\rank E_1}a_2^{\rank E_2}\ldots a_k^{\rank E_k}$.
\end{theorem}
\begin{proof} Now $AE_i = a_iE_i^2 = a_iE_i$. Thus each column of $E_i$ is
 an eigenvector of $A$ corresponding to the eigenvalue $a_i$. Thus there
 are at exist $\rank E_i$ linearly independent eigenvectors corresponding to
 the eigenvalue $a_i$. Since $\rank E_1 + \rank E_2 + \ldots + \rank
 E_k = n$, there are exactly $\rank E_i$ linearly independent
 eigenvectors corresponding to the eigenvalue $a_i$. 
Let $r_i =\rank E_i$. Let these $r_i$
linearly independent eigenvectors corresponding to $a_i$ be
 denoted by $v_{i,1}, v_{i,2}, \ldots v_{i,r_i}$. Do this for each $i$.

Any  column of $E_i$ is perpendicular to any column of $E_j$ for
$i\neq j$ as  $E_iE_j^* = 0$.

Suppose now $\sum_{j=1}^{r_1} \alpha_{1,j} v_{1,r_j} +
 \sum_{j=1}^{r_2}\alpha_{2,j}v_{2,r_j} + \ldots +
 \sum_{j=1}^{r_k}\alpha_{k,j}= 0$.

Multiply through by $E_s$ for $1\leq s \leq k$. This gives
 $\sum_{j=1}^{r_k}\alpha_{k,j}v_{k,j} = 0$ from which it follows that
 $\alpha_{k,j} = 0$ for $j=1,2, \ldots r_k$.  

Thus the set of vectors $S=\{v_{1,1}, v_{1,2}, \ldots v_{1,r_1},
 v_{2,1}, v_{2,2},
 \ldots, v_{2,r_2} \ldots, \ldots, v_{k,1}, v_{k,2}, \ldots, v_{k,r_k}\}$
 is linearly independent and form a basis for $F^n$ -- remember that
 $\rank (E_1+E_2+\ldots + E_k)= n$. Hence $A$ can be
 diagonalised by the matrix of these vectors and thus there is a
 non-singular matrix $P$ such that 
$P^{-1}AP = D$ where $D$ is a diagonal matrix consisting of 
 the $a_i$ repeated $r_i$ times for each $i=1,2, \ldots k$.

Hence $|A| = |D| = a_1^{r_1}a_2^{r_2}\ldots a_k^{r_k}$.
\end{proof}




\end{document}